\documentclass[12pt]{article} 
 \usepackage{geometry}
\usepackage[sectionbib]{natbib}
\usepackage{natbib}
\usepackage{array,epsfig,fancyheadings,rotating}
\usepackage{sectsty}
\usepackage{setspace}
\usepackage{titlesec}
\usepackage{algorithm}
\usepackage{algpseudocode}
\usepackage[]{hyperref}  
\usepackage{etoolbox}
\usepackage[font=small]{caption}
\usepackage{amsmath}
\usepackage{amssymb}
\usepackage{amsfonts}
\usepackage{multirow}
\usepackage{amsthm}
\usepackage{adjustbox}
\usepackage{thm-restate}
\usepackage{enumitem}
\usepackage{bbm}
\usepackage{amsmath,amssymb,bm,mathrsfs,mathtools}

\newtheorem{theorem}{Theorem}[section]

\newtheorem{condition}{Condition}[section]
\newtheorem{assumptions}[theorem]{Assumption}
\newtheorem{proposition}[theorem]{Proposition}
\newtheorem{lemma}[theorem]{Lemma}
\newtheorem{corollary}[theorem]{Corollary}
\newtheorem{remark}[theorem]{Remark}
\theoremstyle{definition}

\newtheorem{example}[theorem]{Example}

\def\bz {{\bf z}}

\def\bP{{\mathbbm{P}}}

\newcommand{\real}{\mathbb{R}}

\date{}

\allowdisplaybreaks
\title{A conformal test of linear models via permutation-augmented regressions}
\author{Leying  Guan\thanks{Dept. of  Biostatistics,
    Yale University, leying.guan@yale.edu}
    }
\begin{document}
\maketitle
\begin{abstract}
Permutation tests are widely recognized as robust alternatives to tests based on normal theory. Random permutation tests have been frequently employed to assess the significance of variables in linear models. Despite their widespread use, existing random permutation tests lack finite-sample and assumption-free guarantees for controlling type I error in partial correlation tests. To address this ongoing challenge, we have developed a conformal test through permutation-augmented regressions, which we refer to as PALMRT. PALMRT not only achieves power competitive with conventional methods but also provides reliable control of type I errors at no more than $2\alpha$, given any targeted level $\alpha$, for arbitrary fixed designs and error distributions. We have confirmed this through extensive simulations.

Compared to the cyclic permutation test (CPT) and residual permutation test (RPT), which also offer theoretical guarantees, PALMRT does not compromise as much on power or set stringent requirements on the sample size, making it suitable for diverse biomedical applications. We further illustrate the differences in a long-Covid study where PALMRT validated key findings previously identified using the t-test after multiple corrections, while both CPT and RPT suffered from a drastic loss of power and failed to identify any discoveries. We endorse PALMRT as a robust and practical hypothesis test in scientific research for its superior error control, power preservation, and simplicity. An R package for PALMRT is available at \url{https://github.com/LeyingGuan/PairedRegression}.
\end{abstract}

\section{Introduction}

Consider a linear regression model
\begin{equation}
\label{eq:model1}
y_i = x_i\beta+\bz_i^\top\theta+\epsilon_i, i = 1,\ldots,n
\end{equation}
with features $x_i\in \real$ and $\bz_i\in \real^p$, and random errors $\varepsilon_i$  independent of $(x_i, \bz_i)$.  Testing whether the coefficient \(\beta\) is zero in this linear model, or equivalently, examining the partial correlation between \(x\) and \(y\), remains a fundamental statistical query and a prevalent approach in applications like biological signature discovery.  For instance, a key question in the recent MY-LC study is whether long-COVID (LC) is associated with specific cell type proportions, after adjusting for age, sex, and Body Mass Index (BMI)  \citep{klein2022distinguishing} :
\begin{equation}
\label{eq:MYLC}
\mbox{cell type frequency}\sim \mbox{intercept}+\mbox{LC}+\mbox{age}+\mbox{sex}+\mbox{BMI}+\mbox{age}\times\mbox{BMI} +\mbox{sex}\times \mbox{BMI}.
\end{equation}
While the F/t-test is a standard approach \citep{fisher1922goodness,fisher1924036, fisher1970statistical}, it may yield anti-conservative p-values under ill-behaved error distributions or limited sample sizes, compromising the reliability of scientific discoveries. Therefore, there is a need for a valid hypothesis test that minimizes assumptions on the noise distribution and does not rely on asymptotic theory. Motivated by this, we seek a robust and straightforward testing procedure for partial correlation under the sole assumption of exchangeability.

\begin{assumptions}
\label{ass:exchangeble_noise}
The noise variables $\varepsilon_1,\ldots, \varepsilon_n$ are exchangeable with in law.
\end{assumptions}

The problem of testing partial correlations has been intensively studied, with a focus on developing various random permutation methods to increase robustness against diverse noise distributions.  \citet{draper1966testing} introduced a permutation technique for partial correlation testing, referred to as PERMtest, which involves permuting $x$ alone and disrupts the relationship between $x$ and $Z$. Subsequent methods better accounted for this relationship. A family of methods, such as the Freedman and Lane test (FLtest) \citep{freedman1983nonstochastic} and the Kennedy-test \citep{kennedy1995randomization}, permutes response residuals obtained from a reduced model that regresses $y$ on $Z$. Another type of method, such as the  Braak-test \citep{ter1992permutation}, shuffles residuals from a full model regressing $y$ on both $x$ and $Z$. These methods have demonstrated empirically robust type I error control in various benchmark studies when using pivotal test statistics like t-test or F-test statistics  \citep{hall1991two,westfall1993resampling, anderson2001permutation, winkler2014permutation}. However, finite-sample and assumption-free theoretical guarantees for these random permutation tests remain elusive. This contrasts with permutation tests for simple correlation, which are theoretically sound in terms of type I error control \citep{pitman1937significance,manly2006randomization, edgington2007randomization}. Indeed, even the widely-examined and  recommended FLtest can, under specific conditions, yield drastically inflated type I error rates, as we demonstrate later.

Recently, \citet{lei2021assumption} introduced the cyclic permutation test (CPT), which offers worst-case guarantees for controlling type I error. Given a sample size \( n \) and a total feature dimension \( p \) for \( z \), CPT theoretically ensures type I error control at a target level \( \alpha \), provided \( n > (\frac{1}{\alpha} - 1)p \). This condition is often challenging to meet in biomedical studies. For instance, in immunology research, sample sizes frequently hover around 100 or a few hundreds, while investigating simultaneously hundreds or even more different biomarkers for their partial correlations with a primary feature of interest. Although applying CPT with a \( \alpha = 0.05 \) cutoff may suffice with fewer concomitant covariates, the situation complicates when multiple hypothesis corrections are applied, as smaller nominal \( p \)-value thresholds are needed to maintain a reasonable false discovery rate (FDR).  In an independent pursuit of robust partial correlation analysis with increased power and reduced sample size, \citet{wen2022residual} introduces the residual permutation test (RPT), defined under the condition $n > 2p$. This test utilizes a sequence of specially-designed permutation matrices, which, in conjunction with the identity matrix, form a group. However, there's a trade-off in the choice of the size parameter $K$: a small permutation set size limits testing for small p-values, while a large $K$ reduces RPT power significantly by design (see Table \ref{tab:summary}). Obtaining a sequence of well-designed permutation matrices with a reasonable size, such as $(\lceil \frac{1}{\alpha}\rceil-1)$ as suggested by \citet{wen2022residual}, remains challenging for small $\alpha$. Following the proposed RPT Algorithm \cite{wen2022residual}, a non-trivial test  requires $\alpha > \frac{1}{n}$.

In this manuscript, we develop PALMRT, which stands for Permutation-Augmented Linear Model Regression Test, as a conformal test to examine the partial correlation relationships in a linear model. Unlike traditional random permutation methods, PALMRT not only empirically controls type I error but also theoretically guarantees a worst-case coverage of \(2\alpha\) at any targeted error level \(\alpha\). It offers well-calibrated \(p\)-values for finite-sample type I error control under arbitrary fixed designs or noise distributions.  Table~\ref{fig:comparison_summary} provides a brief summary of different methods regarding their construction, empirical performance,  theoretical guarantee, and dimension constraints for non-trivial rejection under the i.i.d random Gaussian design.

In our empirical analyses, PALMRT not only maintains empirical Type I error rates below the designated thresholds across diverse simulation settings but also exhibits comparable power to established methods like FLtest when \( (n\slash p) \) is moderately large. It significantly outperforms CPT  and RPT in commonly encountered scenarios. Upon re-analyzing the MY-LC study, PALMRT validates the top findings from the original study, while CPT and RPT show reduced discoveries and fail to confirm any main findings after multiple test corrections. Consequently, we recommend adopting PALMRT as the default robust procedure for detecting significant partial correlations in our daily research.

In our empirical analyses, PALMRT not only maintains empirical Type I error rates below the designated thresholds across diverse simulation settings but also exhibits comparable power to established methods like FLtest when \( (n\slash p) \) is moderately large and outperforms CPT in commonly encountered scenarios. Upon re-analyzing the MY-LC study, we confirmed the robustness of the top findings from the original study using PALMRT.  Consequently, we advocate for the adoption of PALMRT as a robust default procedure for detecting significant partial correlations in our daily  research.

 \begin{table}
 \caption{Summary of representative existing permutation methods and PALMRT. Let \( H^z \) and \( H^{xz} \) represent the projection matrices onto the column spaces of the concomitant features \( Z \) and all features \( (x, Z) \), respectively. Given a permutation \( \pi \) of \( (1, \ldots, n) \), \( x_{\pi} \) and \( Z_{\pi} \) denote the permuted versions of \( x \) and \( Z \). The intercept in the linear model is omitted for brevity. The columns labeled ``Original Model" and ``Permuted Model" detail the models used to generate the original and permuted test statistics, respectively. For CPT, \( \eta \) is a length-\( n \) vector used to define valid cyclic permutations and generate $m$ cyclic permutation copies. Define $L = \lfloor \frac{n}{m}\rfloor$ as the cyclic step size. CPT requires $\eta^\top Z_{(kL+1):(kL+n)}$  to be a constant vector for all $0\leq k\leq m$, where a  index greater than $n$ is looped back from the begging, e.g., $x_{kL+n}=x_{kL}$. Let $V\in \real^{n\times (n-p)}$ be an orthonormal matrix orthogonal to the column space of $Z$.  RPT generates $m$ permutations $\pi_1, \ldots, \pi_m$ based on a variant of cyclic permutation to satisfy the group requirement, resulting in $V_{\pi_k}$ as the row-permuted version of $V$ for $k=1, \ldots, m$. RPT then constructs $\tilde V_k \in \real^{n\times (n-2p)}$ as the orthonormal matrix belonging to the intersection of the column space of $V$ and $V_{\pi_k}$.  The ``Empirical" column indicates empirical type I error control from previous studies (statement on PALMRT is based on our study). The column ``Worst-case"  indicates the finite-sample and assumption-free theoretical guarantee  for a given level \( \alpha \), and ``$\gg \alpha$" represents that no guarantee has been established.  The column ``Restriction" specifies the minimum sample size when a test is defined and can be non-trivial under an i.i.d Gaussian design. Except CPT and RPT, F-statistics are utilized for comparison. }
  \label{tab:summary}
  \label{fig:comparison_summary}
 \begin{center}
 \begin{adjustbox}{width = \textwidth}
  \begin{tabular}{|l|l|l|l|l|l|}
   \hline
 Method & Original Model & Permuted Model & Empirical & Worst-Case & Restriction\\\hline
 PERMtest&$y\sim x+Z$&$y\sim x_{\pi}+Z$& $\lesssim\alpha$& $\gg\alpha$& $n>p$\\\hline
 FLtest&$y\sim x+Z$&$[(I-H^z)y]_{\pi}\sim x+Z$&$\lesssim\alpha$& $\gg\alpha$& $n>p$\\\hline
Kenny-test&$y\sim x+Z$&$[(I-H^z)y]_{\pi}\sim (I-H^z)x$& $\lesssim\alpha$& $\gg\alpha$& $n>p$\\\hline
Braak-test&$y\sim x+Z$&$H^{xz}y+[(I-H^{xz})y]_{\pi}\sim x+Z$& $\lesssim\alpha$& $\gg\alpha$& $n>p+1$\\\hline
  CPT&$\eta^\top y$& $\eta^\top y_{(kL+1):(kL+n)}$&$\approx\alpha$&$\leq \alpha$&$\frac{n}{p}> \frac{1}{\alpha}-1$\\\hline
RPT&$\min\limits_{\tilde V\in \{\tilde V_1,\ldots, \tilde V_m\}}|x^\top \tilde V \tilde V^\top y|$& $|x^\top \tilde V_k \tilde V^\top_k y_{\pi_k}|$&$\lesssim\alpha$&$\leq \alpha$&$n >\max( 2p, \frac{1}{\alpha})$\\\hline 
PALMRT&$y\sim x+Z+Z_{\pi}$&$y\sim x_{\pi}+Z_{\pi}+Z$& $\lesssim\alpha$& $\leq 2\alpha$& $n>2p$\\\hline
 \end{tabular}
  \end{adjustbox}
 \end{center}
 \end{table}

\section{Conformal test via permutation-augmented regressions}
\subsection{Construction of PALMRT}
Let \( x\in \real^n \) be the target feature, and \( Z \) the observation matrix with rows \( \bz_1, \ldots, \bz_n \). Define \( [n] \) as the vector \( (1, \ldots, n) \) and \( \pi \) as a permutation of \( [n] \). Denote \( x_{\pi} \) and \( Z_{\pi} \) as the row-permuted versions of \( x \) and \( Z \) respectively. PALMRT is a random permutation method for testing partial correlation. For each randomly generated permutation $\pi$ of $[n]$, it constructs ``original" and  ``permuted test" statistics based on a pair of permutation-augmented regressions:
\begin{itemize}
\item Original test statistic $T_{original}$:  the F-statistic for significance of $x$ in  the model $y\sim x+Z+Z_{\pi}$:
\[
T_{original} = \left(\|(I-H^{zz_{\pi}})y\|_2^2-\|(I-H^{xzz_{\pi}})y\|_2^2\right)\slash \|(I-H^{xzz_{\pi}})y\|_2^2.
\]

\item Permuted test statistic $T_{perm}$: the F-statistic  for significance of $x_{\pi}$ in the model $y\sim x_{\pi}+Z+Z_{\pi}$:
\[
T_{perm} = \left(\|(I-H^{zz_{\pi}})y\|_2^2-\|(I-H^{x_{\pi}zz_{\pi}})y\|_2^2\right)\slash \|(I-H^{x_{\pi}zz_{\pi}})y\|_2^2.
\]
\end{itemize}
Here, \( H^{*} \) denotes the projection matrix onto the column space of its argument \( * \). For instance, \( H^{z_{\pi}z} \) represents the projection matrix onto the column space of \( (Z_{\pi}, Z) \),  \( H^{xzz_{\pi}} \) onto that of \( (x, Z, Z_{\pi}) \), etc.  We have more confidence of having non-zero $\beta$ if $T_{original}$ is larger than $T_{perm}$. One can easily verify that comparing  \(T_{\text{original}} \) and \( T_{\text{perm}} \) is equivalent to comparing the following simplified expressions on the fitted residuals:
\begin{equation}
\label{eq:PREGstat0}
T_{original}=\|(I-H^{x_{\pi}zz_{\pi}})y\|_2^2,\quad T_{perm} =\|(I-H^{xzz_{\pi}})y\|_2^2.
\end{equation}
We adopt eq.~\eqref{eq:PREGstat0} to construct the original and permuted statistics for any given permutation. Subsequently, we generate \( B \) random permutations \(\{ \pi_b \}^{B}_{b=1}\) of \( [n] \) uniformly and compare the associated original and permuted statistics to compute the p-value for testing \( H_0: \beta = 0 \). The complete procedure is outlined in Algorithm~\ref{alg:test}. 
\begin{algorithm}
\caption{A conformal test for partial correlation via paired regression}\label{alg:test}
\begin{algorithmic}[1]
\Procedure{PALMRT}{$y, x, Z, B$}\quad\;\;\;\Comment{Test for partial correlation $y\sim x|Z$ using $B$ permutations}
\For{$b=1,\ldots, B$}
\State Generate a random permutation $\pi_b$.
\State Construct a pair of statistics as in eq.\;(\ref{eq:PREGstat0}) with $\pi\leftarrow\pi_b$:
\[
(T_{0b},T_{b0})\leftarrow(T_{original}, T_{perm})
\]
\State Set $\omega_{b}=\frac{1}{2}\mathbbm{1}\{T_{b0}=T_{0b}\}+ \mathbbm{1}\{T_{b0}>T_{0b}\}$.
\EndFor
\State Construct p-value for $H_0: \beta = 0$ as $p_{val} \leftarrow \frac{1+\sum_{b=1}^B \omega_{b}}{B+1}$.
\State \textbf{return} $p_{val} $
\EndProcedure
\end{algorithmic}
\end{algorithm}

\subsection{An example differentiating FLtest and PALMRT}
We offer an illustrative example, see Example~\ref{example1}, to demonstrate PALMRT's superior robustness compared to FLtest. Although FLtest has been lauded for its type I error control in prior studies, it may fail under extreme noise and design configurations. In contrast, PALMRT consistently maintains type I error control.
\begin{example}
\label{example1}

We set $n = 100$ and $p=1$, and examine a special design  where $x=(1,0,\ldots, 0)^\top$, $z = (0,1,0,\ldots,0)^\top$.  We generate the response $y$ under the global null as $y\sim \varepsilon$, where $\varepsilon \sim N(0, I_{n\times n}) + 10^4 \times \text{Multinomial}(1;\frac{1}{n}, \ldots, \frac{1}{n}) \times (-1)^{\text{Bernoulli}(\frac{1}{2})}$, denoted as ``multinomial noise."  This extreme noise scenario serves as a robustness test. The feature design represents an imbalanced ANOVA setup with many control samples but only one sample per treatment group. Even though exact permutation is feasible by shuffling the zero rows, we employ standard FLtest and PALMRT to evaluate their empirical coverage using 2000 random draws of $y$. Figure~\ref{fig:counter_example}A shows the miscoverage ratio, defined as 
\[
\text{miscoverage ratio} = \frac{\text{Empirical miscoverage}}{\text{Targeted miscoverage}}.
\]
As the targeted level $\alpha$ becomes small, FLtest encountered excessively inflated type I error -- more than 10-fold that of the targeted level when $\alpha = 0.001$. In contrast, PALMRT controls type I error for small $\alpha$. Figures~\ref{fig:counter_example}B-C  are the histograms of p-values using FLtest and PALMRT.  The distribution of p-values from FLtest contains a spike close to 0, leading to its type I error inflation.
\begin{figure}
\includegraphics[width = 1\textwidth]{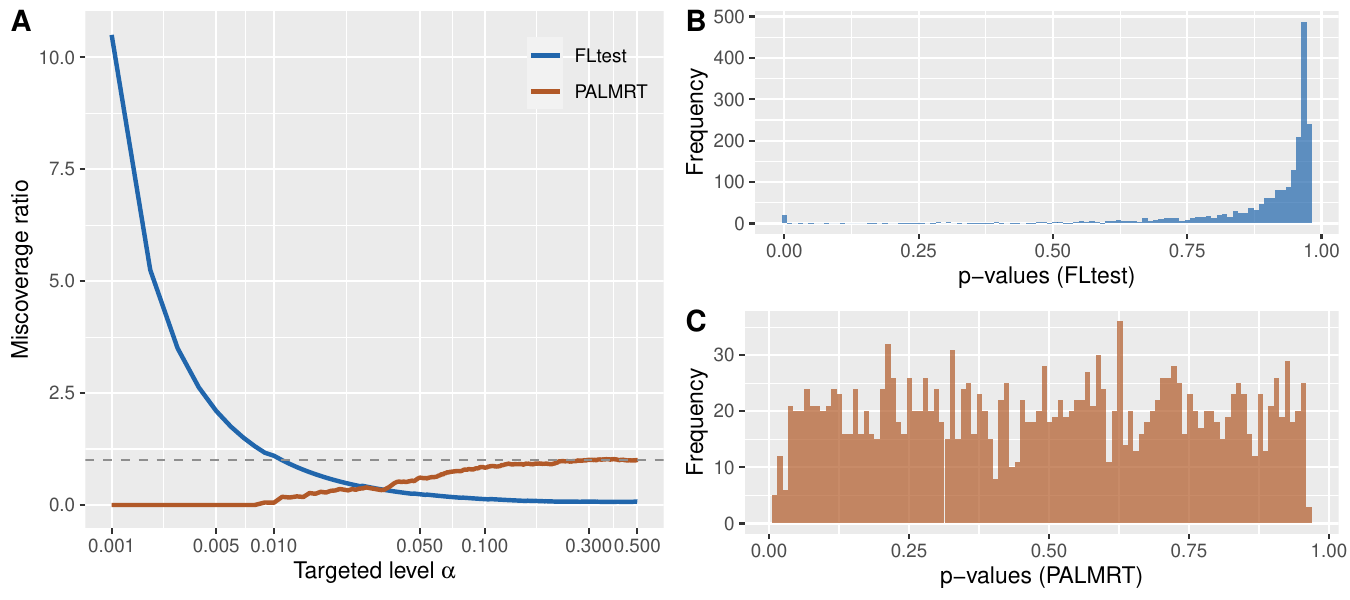}
\caption{Comparison between FLtest and PALMRT on a distinguishing example. Panel A shows the miscoverage ratios using FLtest and PALMRT in Example \ref{example1} as we range $\alpha$ from 0.001 to 0.5. Panels B-C display the histograms of p-values using the two tests.}
\label{fig:counter_example}
\end{figure}
\end{example}
Example \ref{fig:counter_example} is a distinguishing example where FLtest suffers from being severely anti-conservative for small  $\alpha$ but PALMRT continues to offer robust type I error control. This robustness of PALMRT is universally true. In the next section, we will show that PALMRT, as well as a family of other permutation tests based on paired constructions, guarantees a maximum type I error of \(2\alpha\), irrespective of the noise distribution, design, and sample size.
\section{Type I error control guarantee}
In this section, we establish that PALMRT provides finite-sample type I error guarantees under any design and exchangeable noise. This property is generalized to a family of tests that performs comparisons with  paired statistics as realizations of a specially-designed bi-variate function on permutation orders.

Under the null hypothesis $H_0$, PALMRT essentially compares a bi-variate function across two given permutations and their swaps. Specifically, for any two permutations \(\pi_1\) and \(\pi_2\) of \([n]\), we define a bi-variate function, which also incorporates data \(x\), \(Z\), and unobserved noise \(\varepsilon\) as model parameters, as follows:
\begin{equation}
\label{eq:PREGstats}
T^{PALMRT}(\pi_1, \pi_2; x, Z, \varepsilon)=\|(I-H^{x_{\pi_2}z_{\pi_2}z_{\pi_1}})\varepsilon\|_2^2.
\end{equation}
Let $\pi_0$ be the identity permutation of $[n]$.
\begin{proposition}
\label{prop:function_PREG}
Under the null hypothesis \(H_0\), the statistic pair \((T_{0b},T_{b0})\) are realizations of \(T^{PALMRT}(.,.;x, Z, \varepsilon)\) evaluated at \((\pi_0, \pi_b)\) and \((\pi_b, \pi_0)\)  respectively:
\[
(T_{0b}, T_{b0}) = (T^{PALMRT}(\pi_0, \pi_b), T^{PALMRT}(\pi_b, \pi_0)).
\]
\end{proposition}
Many existing random permutation tests, including FLtest and PERMtest, can be expressed as realizations of such bi-variate functions. For example, construct bi-variate functions $T^{PERM}(.)$ and $T^{FL}(.)$ defined below,
\begin{align*}
&T^{PERM}(\pi_1, \pi_2; x, Z, \varepsilon)=\frac{\|(I-H^{z})\varepsilon\|_2^2-\|(I-H^{x_{\pi_1}z})\varepsilon\|_2^2}{\|(I-H^{x_{\pi_1}z})\varepsilon\|_2^2\slash (n-p-2)},\\
&T^{FL}(\pi_1, \pi_2; x, Z, \varepsilon)=\frac{\|(I-H^{z_{\pi_1}})(I-H^{z})\varepsilon\|_2^2-\|(I-H^{x_{\pi_1}z_{\pi_1}}(I-H^{z})\varepsilon\|_2^2}{\|(I-H^{x_{\pi_1}z_{\pi_1}})(I-H^{z})\varepsilon\|_2^2\slash (n-p-2)}.
\end{align*}
Let \(T_0\) denote the original statistic and \(\{T_b\}_{b=1}^B\) represent \(B\) permuted test statistics. Under \(H_0\), it can be verified via direct calculation that  \( (T_{0}, T_{b}) \) are realizations at \( (\pi_0, \pi_b) \) and \( (\pi_b, \pi_0) \) of \( T^{PERM}(.,.;x, Z,\varepsilon) \) or \( T^{FL}(.,.;x, Z,\varepsilon) \) for PERMtest and FLtest respectively, with the second argument in the bivariate functions being inactive.

What sets \( T^{PALMRT}(...) \) apart and enables its theoretical guarantee? The crucial distinction between \( T^{PALMRT}(...) \) and \( T^{FL}(...) \) or \( T^{PERM} (...)\) lies in the transferability of permutations from the noise parameter \( \varepsilon \) to its permutation arguments. This property holds uniquely for \( T^{PALMRT}(...) \) across all noise realizations and designs. Let \( \sigma \) be an arbitrary permutation of \( [n] \) and \( \sigma^{-1} \) its inverse, such that \( \sigma \circ \sigma^{-1} = \sigma^{-1} \circ \sigma = \pi_0\) with \( \circ \) denoting composition. Then, any permutation of the parameters \( \varepsilon \) in  \( T^{PALMRT}(...) \) can be expressed equivalently as applying the inverse permutation $\sigma^{-1}$ to the permutation arguments \( \pi_1 \) and \( \pi_2 \).
\begin{proposition}
\label{prop:transfer}
The application of a permutation $\sigma$ to $\varepsilon$ is equivalent to applying the permutation $\sigma^{-1}$ to $\pi_1$, $\pi_2$ in $T^{PALMRT}$:
\[
T^{PALMRT}(\pi_1, \pi_2; x, Z, \varepsilon_{\sigma}) = T^{PALMRT}(\pi_1\circ \sigma^{-1}, \pi_2\circ \sigma^{-1}; x, Z, \varepsilon).
\]
\end{proposition}
Proposition \ref{prop:transfer} is derived from simple term rearrangement, and we omit its proof here. This transferability property is pivotal for establishing the type I error guarantees. In fact, for any paired statistics \( T_{0b}, T_{b0} \) which can be considered as realizations of a bi-variate function \( T(.,.;x, Z, \varepsilon) \) at \( (\pi_0, \pi_b) \) and \( (\pi_b, \pi_0) \) under \( H_0 \), the resulting p-value from comparing $T_{0b}$ to  $T_{b0}$ offers  a theoretical guarantee as long as \( T(.,.;x, Z, \varepsilon) \)  satisfies the transferability condition \ref{cond:transfer}, as outlined in Theorem \ref{thm:main}.
\begin{condition}
\label{cond:transfer}
For any permutations $\pi_1, \pi_2, \sigma$ of $[n]$, the function $T(.,.; x, Z, \varepsilon)$ satisfies
\[
T(\pi_1, \pi_2; x, Z, \varepsilon_{\sigma}) = T(\pi_1\circ \sigma^{-1}, \pi_2\circ \sigma^{-1}; x, Z, \varepsilon).
\]
\end{condition}
\begin{theorem}
\label{thm:main}
Let \( \pi_1, \ldots, \pi_B \) be \( B \) uniformly random permutations of $[n]$, and \( T(., .; x, Z, \varepsilon) \) be a bi-variate function satisfying condition \ref{cond:transfer}. Under the null hypothesis \( H_0 \), construct paired statistics \( (T_{0b}, T_{b0}) \) as
\[
T_{0b} = T(\pi_0, \pi_b; x, Z, \varepsilon), \quad T_{b0} = T(\pi_b, \pi_0; x, Z, \varepsilon).
\]
Substituting these into Algorithm \ref{alg:test}, we obtain \( \bP_{H_0}[p_{\text{val}} \leq \alpha] < 2\alpha \) for all $\alpha > 0$, where the probability is marginalized over both noise and permutation randomness.
\end{theorem}
\begin{remark}
Interestingly, the empirical version of RPT, discussed independently in \cite{wen2022residual},  is also a realization of paired constructions described in Theorem \ref{thm:main} when $(Z_0, Z_{\pi^{-1}_b})$ is full rank for different permutations $\pi_b$, by setting 
\[
T(\pi_1, \pi_2;x, Z, \varepsilon)=|x_{\pi^{-1}_2}(I-H^{z_{\pi_1^{-1}}z_{\pi^{-1}_2}})\varepsilon|.
\]
This leads to the comparison between $T_{0b}=|x(I-H^{zz_{\pi^{-1}_b}})y|$ and $T_{b0}=|x_{\pi_b^{-1}}(I-H^{zz_{\pi^{-1}_b}})y|$, which can be easily verified. In \cite{wen2022residual}, the authors recognized the power issue with RPT and introduced this empirical version to enhance power. However, it lacks theoretical justification. Here, we demonstrate also that the empirical RPT has a strong theoretical foundation as a special case of Theorem \ref{thm:main} and is a version of PALMRT by replacing the pivotal statistics with residuals inner product, hence, the requirement on permutations by RPT is unnecessary in its context.
\end{remark}
The worst-case bound established in Theorem \ref{thm:main} aligns with the bounds for prediction coverage established for multisplit conformal prediction methods including CV+, Jackknife+, and ensemble conformal predictions \citep{vovk2018cross, barber2021predictive, kim2020predictive, gupta2022nested, han2023conformalized}. However, our focus diverges significantly as we concentrate on hypothesis testing for the true model parameter \( \beta \) rather than an out-of-sample prediction. Traditional exchangeability arguments, applicable when predicting on new samples, are inadequate for assessing the significance of \( \beta \). To tackle this, we employ new arguments exploiting the assumption of exchangeable noise. A proof sketch for Theorem \ref{thm:main} is provided below.
\begin{proof}[Proof sketch of Theorem \ref{thm:main}]
One key insight is that, when considering both the randomness in the permutations and \( \varepsilon \), the \( (B+1) \times (B+1) \) matrix \( T \), with its \( (l, k) \)-th entry as \( T(\pi_l, \pi_k; x, Z, \varepsilon)\), is distributionally equivalent to \( \tilde{T} \) whose \( (l, k) \)-th entry  is \( T(\tilde\pi_l, \tilde\pi_k; x, Z, \varepsilon)\) with $(\tilde\pi_l)_{l=0}^B$ independently and uniformly generated from the permutation space of $[n]$ . This equivalence allows us to analyze the p-value from Algorithm \ref{alg:test} by examining corresponding entries in \( \tilde{T} \).

Define \( f(l, \tilde{T}) = 1 + \sum_{k \neq l} \left( \mathbbm{1}\{\tilde{T}_{kl} > \tilde{T}_{lk}\} + \frac{1}{2} \mathbbm{1}\{\tilde{T}_{kl} = \tilde{T}_{lk}\} \right) \). Then, \( f(0, \tilde{T}) \) corresponds to the numerator when constructing $p_{val}$, the p-value in Algorithm \ref{alg:test}, upon substituting \( T \) with \( \tilde{T} \). Since \( \tilde{\pi}_{l} \) are i.i.d. generated, we expect \( \{f(l, \tilde{T})\}_{l=0}^B \) to be exchangeable for different $l$. Thus, we can bound \( p_{val}\leq \alpha \) by bounding the size of the index set $\{l:f(l, \tilde{T}) \leq \alpha (B+1)\}$, which can be obtained following similar arguments used in proving prediction interval's coverage for multi-split conformal prediction. 
\end{proof}
Theorem \ref{thm:main}  is the main theoretical result of this work, and we include its full proof in Section \ref{sec:proof_main}. Combining Theorem \ref{thm:main} with Proposition \ref{prop:transfer}, we concludes that Algorithm \ref{alg:test} theoretically controls type I error, albeit with a relaxed upper bound of \(2\alpha\).

Of note, unlike existing random permutation tests where the use of pivotal statistics is often crucial, Theorem \ref{thm:main} generalizes beyond pivotal statistics, allowing for other non-pivotal paired constructions. For example, \(T_{0b}, T_{b0}\) could be the absolute value of the regression coefficients of \(x\) and $x_{\pi}$ in the models \(y \sim x + Z + Z_{\pi_b}\) and  \(y \sim x_{\pi} + Z + Z_{\pi_b}\) respectively. For directional tests, the test statistics may employ either the regression coefficients (for positive effects) or their negations (for negative effects).

\section{An Exact Confidence Interval Construction}
In conjunction with the PALMRT \( p \)-value, a confidence interval for \( \beta \) can be constructed by inverting the test. Define $(T_{0b}(\beta), T_{b0}(\beta))$ as the test statistics from replacing $y$ by $(y-x\beta)$ in  Algorithm \ref{alg:CI} when constructing $(T_{0b}, T_{b0})$. Define  \( f(\beta) \) as $f(\beta) = \frac{1 + \sum_{b=1}^B \omega_b(\beta)}{B+1}$,
where $\omega_{b}(\beta) =\mathbbm{1}\{T_{0b}(\beta)< T_{b0}(\beta)\}+\frac{1}{2}\mathbbm{1}\{T_{0b}(\beta) = T_{b0}(\beta)\}$.
\begin{corollary}
\label{corollary:CI}
Set $\rm{CI}_{\alpha} =[\beta_{\min}, \beta_{max}]$ where $\beta_{\min} = \inf\{\beta:f(\beta)> \alpha\}$ and $\beta_{\max} = \sup\{\beta:f(\beta)>\alpha\}$. Then, we have $\min_{\beta} \bP [\beta\in \rm{CI}_{\alpha}]>1-2\alpha$, for all $\alpha > 0$.
\end{corollary}
\begin{remark}
The set obtained through directly inverting \( \{ \beta : f(\beta) > \alpha \} \), is often an interval, but not always guaranteed to be so. By taking the infimum and supremum of this set, we obtain a confidence interval \( CI_{\alpha} \) with worst-case guarantee at least as strong as direct inversion.
\end{remark}
The pertinent question remaining is the efficient computation of the confidence interval \(CI_{\alpha}\) as delineated in Corollary~\ref{corollary:CI}. Traditional methods for constructing the confidence interval in permutation tests typically rely on normal theory, Bootstrap~\citep{efron1987better, efron1994introduction, diciccio1996bootstrap, davison1997bootstrap}, or grid search~\citep{garthwaite1996confidence}. Here, we provide an exact formulation of \(CI_{\alpha}\) via examining critical values, derived from pair-wise comparisons at each permutation $\pi_b$. First, we observe that the contribution from the \(b^{th}\) term to the numerator is explicitly known.
\begin{lemma}
\label{lem:CI_term}
Set \(c_{b1} = \|(I-H^{x_{\pi_b}z_{\pi_b}z})x\|_2^2\), \(c_{b2} = x^\top (I-H^{x_{\pi_b}z_{\pi_b}z})y\), \(c_{b3} = \|(I-H^{x_{\pi_b}z_{\pi_b}z})y\|_2^2\), and \(c_{b4} = \|(I-H^{x_{\pi}z_{\pi_b}z})y\|_2^2\). Then,  we have \(c_{b2}^2\geq c_{b1}(c_{b3}-c_{b4})\) and 
\begin{equation}
\omega_b(\beta) = \left\{
\begin{array}{lll}
\frac{1}{2}\mathbbm{1}\{c_{b3}=c_{b4}\}+\mathbbm{1}\{c_{b3}<c_{b4}\},&\mbox{if}&c_{b1}=0,\\
\frac{1}{2}\mathbbm{1}\{\beta\in [s_b, u_b]\}+\mathbbm{1}\{\beta\in (s_b, u_b)\},&\mbox{if}&c_{b1}>0,\\
\end{array}
\right.
\end{equation}
where \(s_b=\frac{c_{b2}-\sqrt{c_{b2}^2-c_{b1} (c_{b3}-c_{b4}) }}{c_{b1}}\) and \(u_b=\frac{c_{b2}+\sqrt{c_{b2}^2-c_{b1} (c_{b3}-c_{b4}) }}{c_{b1}}\).
\end{lemma}
As a result, we first partition of index set of $\{0,\ldots, B\}$ into three sets  \(A_1=\{b: c_{b1}>0\}\),  \(A_2 = \{b: c_{b1}=0, c_{b3} < c_{b4}\}\), and \(A_3 = \{b: c_{b1}=0, c_{b3} = c_{b4}\}\). The value of \(\omega_b(\beta)\) remains constant as we vary \(\beta\) for \(b\in A_2\cup A_3\). Hence, the requirement of \(f(\beta)>\alpha\) is equivalent to imposing a requirement on \(f_{A_1}(\beta)\) that captures the contribution of \(\omega_b(\beta)\in A_1\), as shown below:
\[
f(\beta) > \alpha \Leftrightarrow f_{A_1}(\beta) = \sum_{b\in A_1} \omega_b(\beta) > (B+1)\alpha - 1 -|A_2|- \frac{1}{2}|A_3| \coloneqq \gamma.
\]
It can be shown that the function value of  \(f_{A_1}(\beta)\) is characterized by comparing \(\beta\) to different \(s_b\slash u_b\) values for \(b\in A_1\):
\begin{equation}
\label{eq:fA}
f_{A_1}(\beta) = \frac{1}{2}\#\{b:s_b \leq \beta\}+ \frac{1}{2}\#\{b:s_b < \beta\}-\frac{1}{2}\#\{b:u_b \leq \beta\}- \frac{1}{2}\#\{b:u_b < \beta\}.
\end{equation}
Let \(t_1 < \ldots < t_M\) denote the ordered values of $M$ unique elements in \(\cup_{b\in A_1} \{s_b, u_b\}\), and let \((m_l^s, m_l^u)\) represent the sizes of \(\#\{b:s_b=t_l\}\) and \(\#\{b:u_b=t_l\}\), respectively. As we increase $\beta$ in\(f_{A_1}(\beta)\), the function value can only changes when we first hit   \(\{t_l\}_{l=1}^M\), or when \(\beta\) slightly increases from these critical values represented. We represent the concept of increasing slightly from these critical values  by  \(\{t_l^+\}_{l=1}^M\), where \(t_l^+\) indicates being infinitesimally larger than \(t_l\).  Using these new quantities introduced, we can re-express \(f_{A_1}(t_1) = \frac{1}{2}(m_{1}^s-m_{1}^{u})\) and identify induction relations  for the function values as we increase $\beta$ to surpass the critical values $t_l$, described as follows:
\begin{equation}
\label{eq:induction_fA}
f_{A_1}(t_{l+1}) =f_{A_1}(t_l^+)+\frac{1}{2}(m_{l+1}^s-m_{l+1}^{u}),\quad f_{A_1}(t_{l}^+) =f_{A_1}(t_l)+\frac{1}{2}(m_{l}^s-m_{l}^{u}).
\end{equation}
We can utilize eq.\;(\ref{eq:induction_fA}) to efficiently calculate all \(f_A(t_l)\) and \(f_A(t^{+}_l)\) in \(O(B)\) time, given \(\{t_l\}_{l=1}^M\) and \((m_l^s, m_l^u)_{l=1}^M\). Acquisition of the $B\times 4$ matrix $(c_{bl})$, \(\{t_l\}_{l=1}^M\) and \((m_l^s, m_l^u)_{l=1}^M\) can be done in \(O(Bnp+B\log B)\) time if we record intermediate quantities from Algorithm \ref{alg:test}. Consequently, efficient comparisons between   \(f_{A}(.)\) with \(\gamma\) to determine \(\beta_{\min}\) and \(\beta_{\max}\) can be achieved.   Algorithm \ref{alg:CI} presents full details of this implementation, and provides exact construction of $CI_{\alpha}$ as stated in Theorem \ref{thm:CI}. 
\begin{algorithm}
\caption{Exact CI construction for PALMRT}\label{alg:CI}
\begin{algorithmic}[1]
\Procedure{CI}{$y, x, Z, B,\alpha$}\quad\;\;\;\Comment{Confidence interval at coverage level $1-\alpha$.}
\State Calculate the $B\times 4$ matrix  $(c_{b1}, c_{b2}, c_{b3}, c_{b4})_{b=1}^B$ as defined in Lemma \ref{lem:CI_term}.
\State Calculate $\gamma$, $\{t_{l}\}_{l=1}^M$,  $(m_l^s, m_l^u)_{l=1}^M$.
\If{$\gamma < 0$}
\State $\rm{CI}_{\alpha}=(-\infty, \infty)$.
\Else
\State Set and record $f_{A_1}(t_1) = \frac{1}{2}(m_1^s - m_1^u)$.
\For{$l=2,\ldots, M$}
\State Calculate and record $f_{A_1}(t_{l-1}^+) $ and  $f_{A_1}(t_{l}) $ as in eq.\;(\ref{eq:induction_fA}).
\EndFor
\If{$\max(f_{A_1}(.))\leq \gamma$}
\State $\rm{CI}_{\alpha}=\emptyset$.
\Else
\State $\beta_{min} = \min\{t_l: f_{A_1}(t_l)\vee f_{A_1}(t_l^+) > \gamma\}$.
 \State $\beta_{max} = \max\{t_l: f_{A_1}(t_l)\vee f_{A_1}(t_{l-1}^+) > \gamma\}$.
 \State $CI_{\alpha}=[\beta_{min}, \beta_{max}]$. 
 \EndIf
\EndIf
\EndProcedure
\end{algorithmic}
\end{algorithm}
\begin{theorem}
\label{thm:CI}
The confidence interval constructed by Algorithm \ref{alg:CI} corresponds to the $CI_{\alpha}$ defined in Lemma \ref{corollary:CI}, and guarantees a worst-case coverage of $(1-2\alpha)$ for a specified mis-coverage level $\alpha$.
\end{theorem}
\begin{remark}
The confidence interval $CI_{\alpha}$ can potentially be an empty set ($\emptyset$). Although this does not invalidate our assertion, an empty set offers limited information in practical contexts and might not be desired. In such situations, we may pt to construct the confidence interval using a normal approximation in the regression $y\sim x+Z$ whenever Algorithm~\ref{alg:CI} produces an empty set as the output for $CI_{\alpha}$.
\end{remark}
Proofs of Lemma \ref{lem:CI_term} and Theorem \ref{thm:CI} are deferred to Appendix~\ref{app:proof} in the Supplement.
\section{Proof of Theorem \ref{thm:main}}
\label{sec:proof_main}
Before conducting numerical experiments comparing PALMRT and existing methods, we provide the full proof to Theorem \ref{thm:main} in this section. We define $\mathcal{S}_n$ as the permutation space of $[n]$, $\mathcal{S}_B$ as the permutation space of $(0,\ldots, B)$, $\mathcal{E}=\{\varepsilon_1,\ldots,\varepsilon_n\}$ as the unordered value set of $(\varepsilon_1,\ldots, \varepsilon_n)$ (duplicates are allowed). Proposition \ref{prop:uniform_permutation} is useful for establishing our exchangeability  statement. Proposition \ref{prop:count_rows} is a minor modification of arguments for bounding ``strange set" used in proving multi-splitting conformal prediction.
\begin{proposition}
\label{prop:uniform_permutation}
Let $\sigma$ and $(\pi_b)_{b=1}^B$ be generated independently and uniformly from $\mathcal{S}_n$. Then $\sigma^{-1}$ and $(\pi_b\circ \sigma^{-1})_{b=1}^B$ are also generated independently and uniformly from $\mathcal{S}_n$.
\end{proposition}
\begin{proof}
First, it is obvious that $\sigma^{-1}$ is  uniformly generated from  since the map $\sigma\mapsto \sigma^{-1}$ is a bijection between  $\mathcal{S}_n$ and itself. Hence, by definition, for any $(B+1)$ permutations $\tau_0, \tau_1, \ldots, \tau_B$ of $[n]$, we have
\begin{align}
&\bP(\sigma^{-1}=\tau_0, \pi_1\circ \sigma^{-1} = \tau_1, \ldots, \pi_B\circ \sigma^{-1} = \tau_B)\notag\\
=&\bP(\sigma^{-1}=\tau_0, \pi_1 = \tau_1\circ \tau_0, \ldots, \pi_B= \tau_B\circ \tau_0)\notag\\
=&\bP(\sigma^{-1}=\tau_0)\bP(\pi_1 = \tau_1\circ \tau_0)\ldots\bP(\pi_B = \tau_B\circ \tau_0)\notag\\
=&(1\slash n!)^{B+1}, \label{eq:prop_unifrom}
\end{align}
where the last two steps have used the fact that  $\sigma^{-1}$ and $(\pi_b)_{b=1}^B$ are independent and uniformly from $\mathcal{S}_n$. Eq~(\ref{eq:prop_unifrom}) is the definition for $\sigma^{-1}$ and $(\pi_b\circ \sigma^{-1})_{b=1}^B$ being independent and uniformly generated from $\mathcal{S}_n$.
\end{proof}
\begin{proposition}
\label{prop:count_rows}
Let $ T$ be any $(B+1)\times (B+1)$ matrix, and  $W$ the $(B+1)\times (B+1)$ comparison matrix where $W_{lk} = \mathbbm{1}\{T_{kl} >T_{lk}\}+\frac{1}{2}\mathbbm{1}\{T_{kl} = T_{lk}\}$ for $l\neq k$ and $W_{ll} = 1$. Let $W_{l.}$ be the $l^{th}$ row sum of  $W$. Then,  for all $\alpha>0$, we have 
\[
|\{l: W_{l.}\leq (B+1)\alpha\}|<2\alpha (B+1).
\]
\end{proposition}
\begin{proof}
Notice that  (1)$W_{lk}\geq 0$ for all $l$, $k$, and (2) $W_{lk}+W_{kl} = 1$ for all $k\neq l$. For any sub-square matrix of $W$ constructed from selecting the same  $m$ columns and rows (denoted the index set as $I_m$), we have
\[
\sum_{l\in I_m}\sum_{k\in I_m} W_{lk} = \frac{(m+1)m}{2}.
\] 
Set $S_{\alpha} =\{l: W_{l.}\leq (B+1)\alpha\}$. Suppose the size of $S_{\alpha}$ is $s\geq 0$. The corresponding $s$ rows in $W$ must implicate:
\begin{align*}
s\alpha(B+1)\geq \frac{s(s+1)}{2}\Rightarrow 0\leq s\leq \max(0, 2\alpha(B+1)-1)<2\alpha (B+1).
\end{align*}
This concludes our proof.
\end{proof}
\subsection{Proof to Theorem~\ref{thm:main}}
\label{app:proof_main_p}
On the one hand, under Assumption \ref{ass:exchangeble_noise} and conditional on the value sets $\mathcal{E}=\{\varepsilon_1,\ldots,\varepsilon_n\}$, generation of $(\pi_1,\ldots, \pi_B, \varepsilon_{\sigma})$ can be characterized by generating \(\sigma, \pi_1, \ldots, \pi_B, \sigma\)  independently and uniformly generated from $\mathcal{S}_n$. 

On the other hand, by Condition \ref{cond:transfer} and denoting $\pi_0$ as the identity permutation of $[n]$, we have
\begin{align}
&\left(T(\pi_0, \pi_1; \varepsilon_{\sigma}),  \ldots, T(\pi_0, \pi_B; \varepsilon_{\sigma}), T( \pi_1, \pi_0; \varepsilon_{\sigma}),  \ldots, T(\pi_B, \pi_0; \varepsilon_{\sigma})\right) \label{eq:proof_main_p1}\\
=&  \left(T(\sigma^{-1}, \pi_1\circ\sigma^{-1}; \varepsilon),  \ldots, T(\sigma^{-1}, \pi_B\circ\sigma^{-1}; \varepsilon), T(\pi_1\circ\sigma^{-1},\sigma^{-1}; \varepsilon),  \ldots, T(\pi_B\circ\sigma^{-1},\sigma^{-1}; \varepsilon)\right). \notag
\end{align}
Here, we have dropped the parameters $x$ and $Z$ in $T(.,.;x, Z, \varepsilon)$ for convenience.  

Write $\tilde\pi_0 = \sigma^{-1}$, $\tilde\pi_1 = \pi_1\circ\sigma^{-1}$, \ldots,  $\tilde\pi_B = \pi_B\circ\sigma^{-1}$, by eq.~(\ref{eq:proof_main_p1}) and Proposition \ref{prop:uniform_permutation}, the marginalized joint distribution of the B pairs of statistics conditional only on $\mathcal{E}$ can be re-expressed equivalently (in distribution) using $(\tilde \pi_b)_{b=1}^B$:
\begin{equation}
\label{eq:proof_main_p2}
(T_{0b}, T_{b0})_{b=1}^B \overset{d}{=} \left(T(\tilde \pi_0, \tilde \pi_b; \varepsilon), T(\tilde \pi_b, \tilde \pi_0; \varepsilon)\right)_{b=1}^B,
\end{equation}
where $\tilde \pi_0,\ldots, \tilde \pi_B$ are independently and uniformly generated from $\mathcal{S}_n$, and $\varepsilon$ can be viewed as fixed. Hence, setting  $(\tilde T_{0b}, \tilde T_{b0}) = \left(T(\tilde \pi_0, \tilde \pi_b; \varepsilon), T(\tilde \pi_b, \tilde \pi_0; \varepsilon)\right)$, to understand the behavior of the constructed p-value using \((T_{0b}, T_{b0})_{b=1}^B\),  we can equivalently consider the distribution of $\widetilde{\rm{p_{val}}}$:
\[
\widetilde{\rm{p_{val}}} = \frac{1+\sum_{b=1}^B\left(\mathbbm{1}\{\tilde T_{b0}>\tilde T_{0b}\}+\frac{1}{2}\mathbbm{1}\{\tilde T_{b0}=\tilde T_{0b}\}\right)}{B+1}.
\]

Now, we complete the full $\tilde T$ matrix by setting $\tilde T_{kl} = T(\tilde \pi_k, \tilde \pi_l; \varepsilon)$ for all $k, l = 0,\ldots, B$. We also set the comparison $W$ matrix of $\tilde T$ as described in Proposition \ref{prop:count_rows} and set  $S_{\alpha}=\{l: W_{l.}\leq \alpha (B+1)\}$. Then, by Proposition \ref{prop:count_rows},  the size of $S_{\alpha}$ is upper bounded and $|S_{\alpha}| < 2\alpha (B+1)$.

Note that the p-value constructed using $\tilde T$ is 
\[
\widetilde{\rm{p_{val}}}=\frac{1+\sum_{b=1}^B\left(\mathbbm{1}\{\tilde T_{b0} > \tilde T_{0b}\}+\frac{1}{2}\mathbbm{1}\{\tilde T_{b0} = \tilde T_{0b}\}\}\right)}{B+1}=\frac{W_{0.}}{B+1}.
\]
To avoid confusion,  we use the lower case $(w_b)_{b=1}^B$ to denote the observed row sums $(W_{b.})_{b=1}^B$ given the current realizations $(\tilde \pi_0,\ldots, \tilde \pi_{B})=(\tau_0, \ldots, \tau_{B})$, and $(W_{b.})_{b=1}^B$  be the random variables as we change the permutations  $(\tilde \pi_0,\ldots, \tilde \pi_{B})$. Then, conditional on the unordered set of permutations $\{\tilde\pi_0,\ldots, \tilde \pi_{B}\}=\{\tau_0,\ldots, \tau_B\}$, the observed permutations take the form $\tilde\pi_b= \tau_{\zeta_b}$, for $b=0,\ldots, B$, where $\zeta$ is a permutation of $(0,\ldots, B)$. Since $\tilde\pi_0,\ldots, \tilde \pi_B$ are i.i.d generated from $\mathcal{S}_n$, $\zeta$ is uniformly generated from the  $\mathcal{S}_B$.

The above results can be used to derive the exchangeability among $W_0,\ldots, W_B$ conditional on $\mathcal{E}$ and  $\{\tilde\pi_0,\ldots, \tilde \pi_{B}\}=\{\tau_0,\ldots, \tau_B\}$. Notice that when $\tilde\pi = \tau_{\zeta}$ is permuted according to $\zeta$, the row sum $W_{0.}=w_{\zeta_0}$ is permuted accordingly:
\begin{align*}
W_0=&1+\sum_{b=1}^B \left(\mathbbm{1}\{T(\tau_{\zeta_b}, \tau_{\zeta_0};\varepsilon)>T(\tau_{\zeta_0}, \tau_{\zeta_b};\varepsilon)\}+\frac{1}{2}\mathbbm{1}\{T(\tau_{\zeta_b}, \tau_{\zeta_0};\varepsilon)=T(\tau_{\zeta_0}, \tau_{\zeta_b};\varepsilon)\}\right)\\
=&1+\sum_{b=1}^B\left(\mathbbm{1}\{T_{\zeta_b,\zeta_0}>T_{\zeta_0,\zeta_b}\}+\frac{1}{2}\mathbbm{1}\{T_{\zeta_b,\zeta_0}=T_{\zeta_0,\zeta_b}\}\right)\\
=&1+\sum_{b\neq \zeta_0}\left(\mathbbm{1}\{T_{b,\zeta_0}>T_{\zeta_0,b}\}+\frac{1}{2}\mathbbm{1}\{T_{b,\zeta_0}=T_{\zeta_0,b}\}\right)\\
=& w_{\zeta_0.},
\end{align*}
Combining the above display with the fact that  $\zeta$ is uniform from $\mathcal{S}_B$, we conclude that $\{W_{b.}\}_{b=0}^B$ are exchangeable. Consequently, we have
\begin{align*}
\bP(W_0\in S_{\alpha})<2\alpha\Rightarrow  \bP\left(\widetilde{\rm{p_{val}}}\leq \alpha\right) < 2\alpha.
\end{align*}
\section{Numerical Experiments}
\label{sec:sim}
We evaluate the performance and type I error control of various methods through numerical experiments, setting the sample size at $n=100$, and varying the dimension of $Z$ in $\{1, 5, 15\}$.   We investigate four designs: i.i.d Gaussian, i.i.d Cauchy, balanced ANOVA where each feature takes roughly an equal number of non-overlapping 1s, and a paired design where each feature assumes a value of 1 at two entries and 0 elsewhere, with one unique and one shared one-valued entry across all features. We also consider three noise settings: Gaussian, Cauchy, and multinomial.  The paired design and multinomial noise distributions are atypical in real-world applications and are included as challenging test cases for evaluating the robustness of FLtest, which has demonstrated commendable empirical performance under more conventional design structures and noise distributions in existing literature.

We compare our proposed PALMRT against six existing methods:(1) F-test, (2) PERMtest, (3) FLtest, (4) CPT with strong pre-odering by the genetic algorithm as per \citet{lei2021assumption}, (5) RPT, and (6) Bias-corrected and Accelerated Bootstrap, previously favored over plain Bootstrap \citep{efron1987better, diciccio1988review, hall1988theoretical, diciccio1996bootstrap}.   All comparisons were conducted at the p-value cut-off $\alpha = 0.05$ in the main paper; PALMRT's type I error control is further explored for $\alpha=0.01$ and $\alpha=0.001$  in supplementary materials.

For random permutation tests (PALMRT, FLtest, PERMtest), we employed F-statistics as the test statistics and set the permutation count $B = 2000$. Bias-corrected and Accelerated Bootstrap and CPT can be  more time-intensive. We used {\it bcaboot} R package with 500 bootstraps and 20 jackknife blocks \citep{efron2020automatic}.   For CPT, we used the implementation from the authors' GitHub, followed the strong ordering approach with 10,000 optimization steps via genetic algorithms, and set the number of cyclic permutations as 19 (corresponding to $\alpha = 0.05$)  \citep{lei2021assumption}.  For RPT, we used the implementation provided by the authors and set the the number of permutation as 19 (corresponding to $\alpha = 0.05$) as recommended \citep{wen2022residual}.

We generated 2,000 independent noise instances $\varepsilon$ for each experimental setting and used them to compute empirical type I error, statistical power, and confidence intervals across various signal-to-noise ratios. Sections \ref{subsec:sim_typeI} and \ref{subsec:sim_power} compare type I error and power among F-test, PERMtest, FLtest, CPT, RPT, and PALMRT. Section \ref{subsec:sim_CI} examines CI coverage for $\beta$ and their median lengths from independent runs using normal theory, Bootstrap and Algorithm \ref{alg:CI}.

\subsection{Type I error control}
\label{subsec:sim_typeI}
\begin{figure}
\includegraphics[width = 1\textwidth]{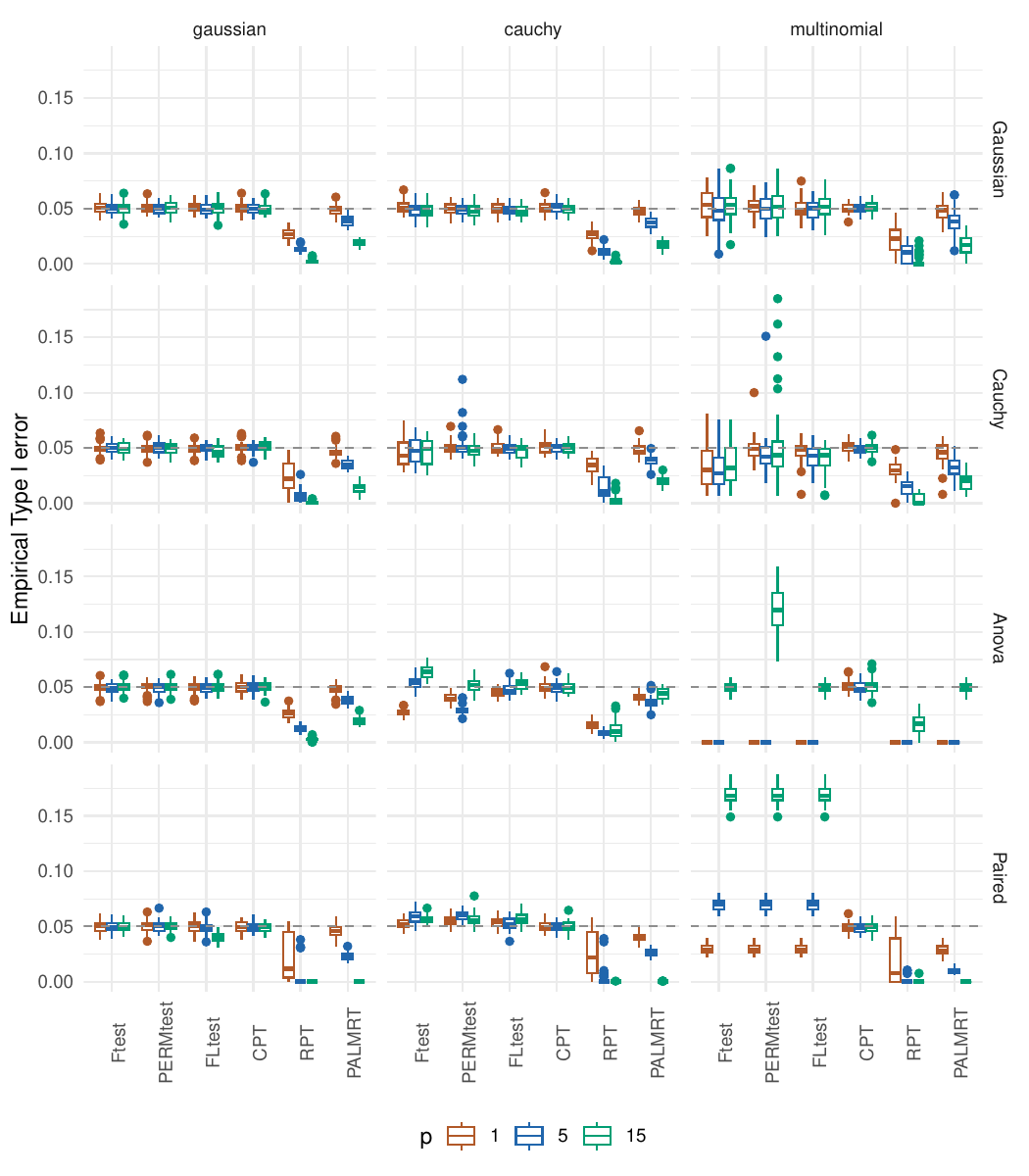}
\caption{Empirical type I error using various methods, organized into boxplots by corresponding designs (row names) and noise distributions (column names). Each boxplot displays the empirical type I error for different methods, separately for different feature dimension $p$ (color). The dashed horizontal line represents the targeted type I error  $\alpha = 0.05$.}
\label{fig:typeIA}
\end{figure}
To empirically assess type I error control, we simulate the global null distribution with $y = \varepsilon$. Figure \ref{fig:typeIA} displays type I errors from 50 independent repetitions for each setting using F-test, PERMtest, FLtest, CPT, RPT and PALMRT.

In the Gaussian noise or the Gaussian\slash Cauchy design contexts, all methods effectively control type I errors. However, for ANOVA or Paired designs, Ftest, FLtest, and PERMtest yield inflated type I errors when noise is non-Gaussian -- especially pronounced in the Paired design with multinomial errors. This underscores that not only F-test, but also FLtest and PERMtest,  lack distribution-free theoretical guarantees and can be anti-conservative in finite-sample settings.

Among the methods tested, only PALMRT, CPT and RPT guarantee worst-case coverage across all experimental settings. Although CPT's accurate type I error control is by design for $n/p \geq 19$,  its accuracy in other settings was somewhat unexpected. For example, in the Gaussian design where  $p=15$, the cyclic invariant constraints of CPT can only be satisfied by  a constant vector $\eta$ with high probability,  leading to a trivial test that always accepts $H_0$. Further examination indicates that these observations arise because cyclic-invariant constraints of CPT are not perfectly satisfied in the implemeted CPT, thereby permitting non-constant $\eta$.  PALMRT demonstrates empirical type I errors close to or below target levels, while RPT, as noted in its original paper, tends to be conservative, with the degree of conservativeness varying based on design, noise distribution, and dimensionality.  Importantly,  the conservative behavior of PALMRT does not sacrifice statistical power when compared to CPT and is generally more powerful compared to RPT, as evidenced by our subsequent power analyses.

\subsection{Power Analysis}
\label{subsec:sim_power}
\begin{figure}
\begin{center}
\includegraphics[width = .9\textwidth, height = 1.2\textwidth]{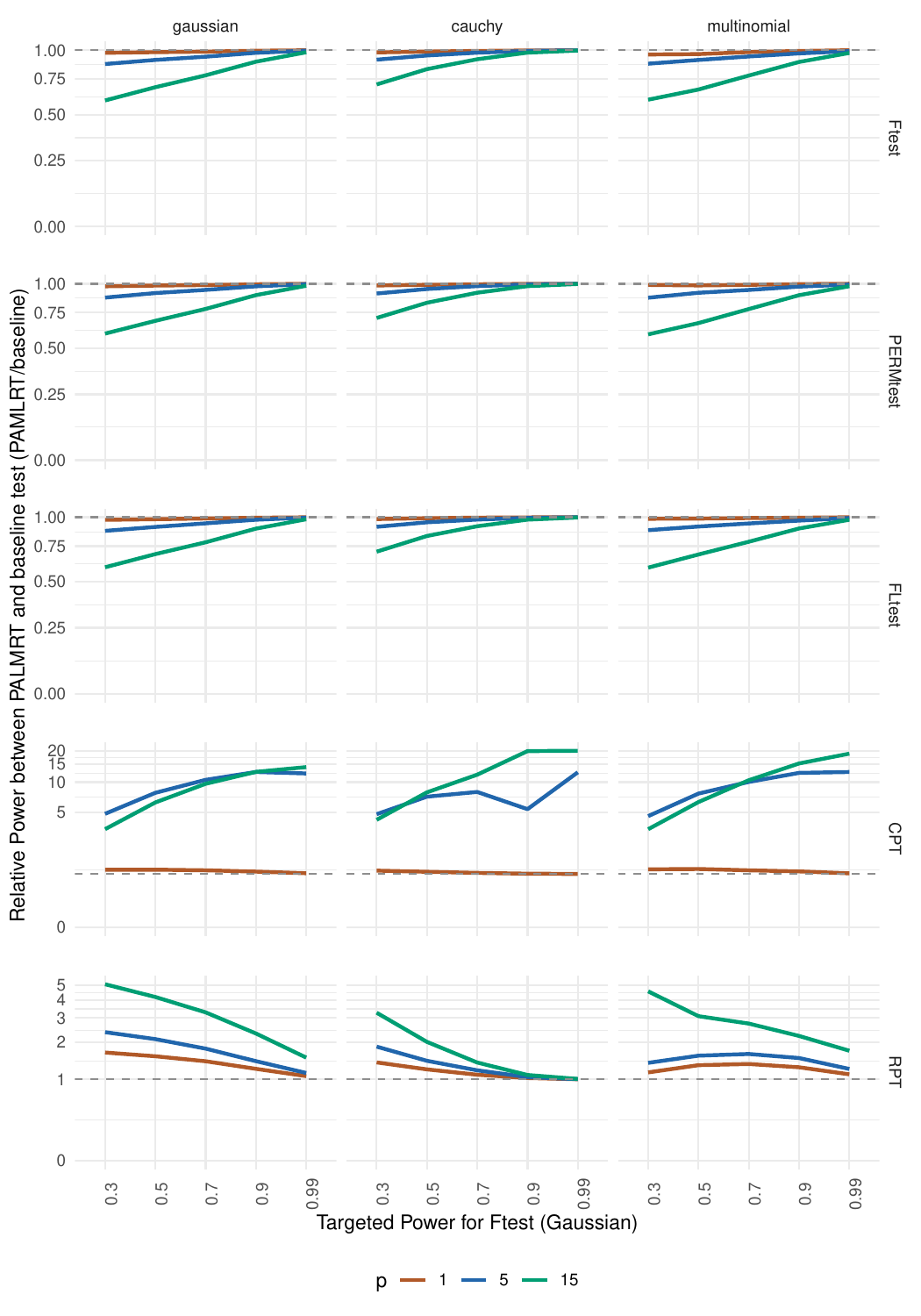}
\caption{Power analysis in Gaussian design, presented as line plots, organized by baseline methods (row names) and noise distributions (column names). Each line plot shows the median power ratio between PALMRT and a baseline method, plotted against the targeted F-test absolute power for various signal sizes. Different colors indicate varying feature dimensions $p$.}
\label{fig:powerA}
\end{center}
\end{figure}

\begin{figure}
\begin{center}
\includegraphics[width = .9\textwidth, height = 1.2\textwidth]{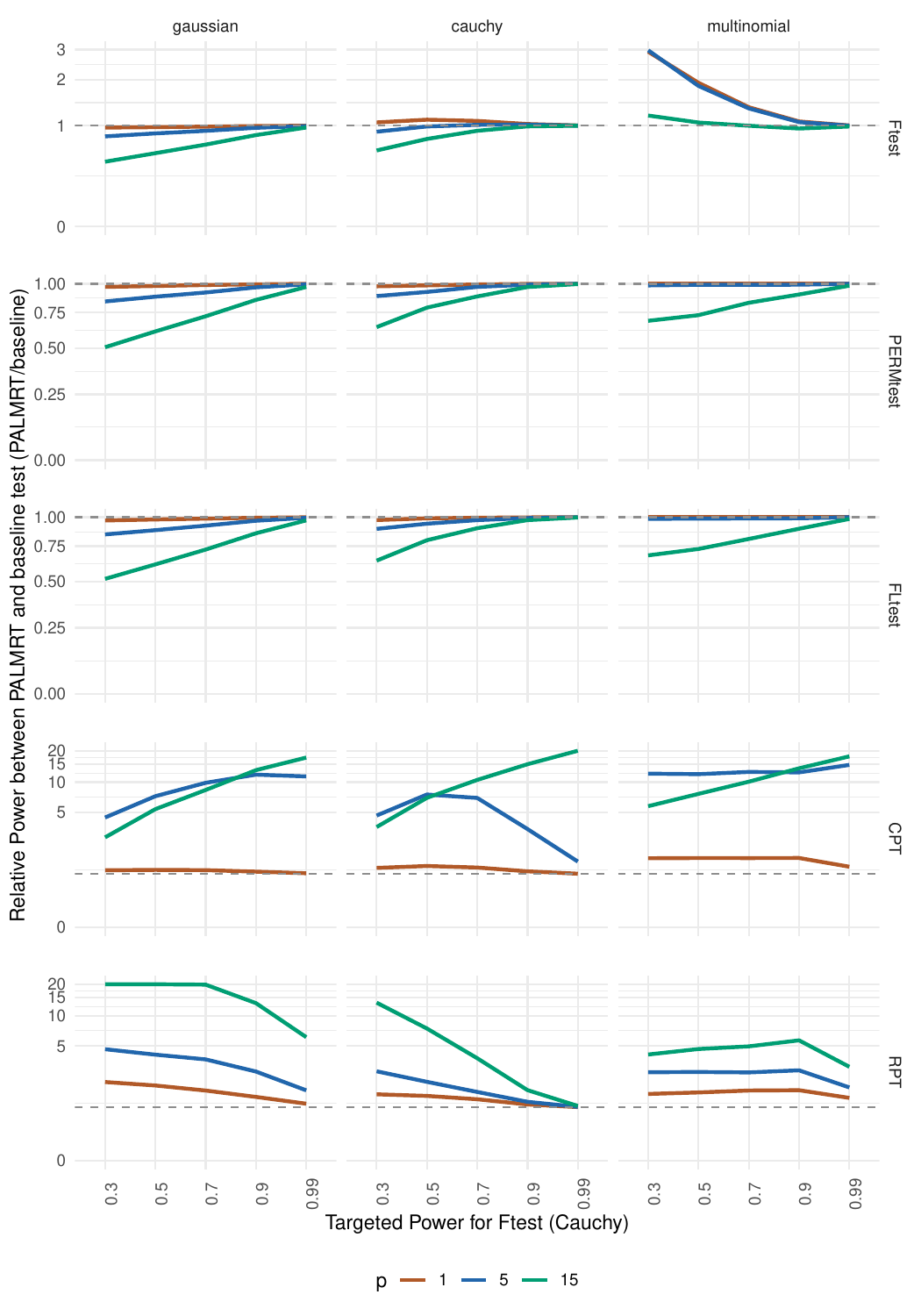}
\caption{Power analysis in Cauchy design, presented as line plots, organized by baseline methods (row names) and noise distributions (column names). Each line plot shows the median power ratio between PALMRT and a baseline method, plotted against the targeted F-test absolute power for various signal sizes. Different colors indicate varying feature dimensions $p$.}
\label{fig:powerB}
\end{center}
\end{figure}

\begin{figure}
\begin{center}
\includegraphics[width = .9\textwidth, height = 1.2\textwidth]{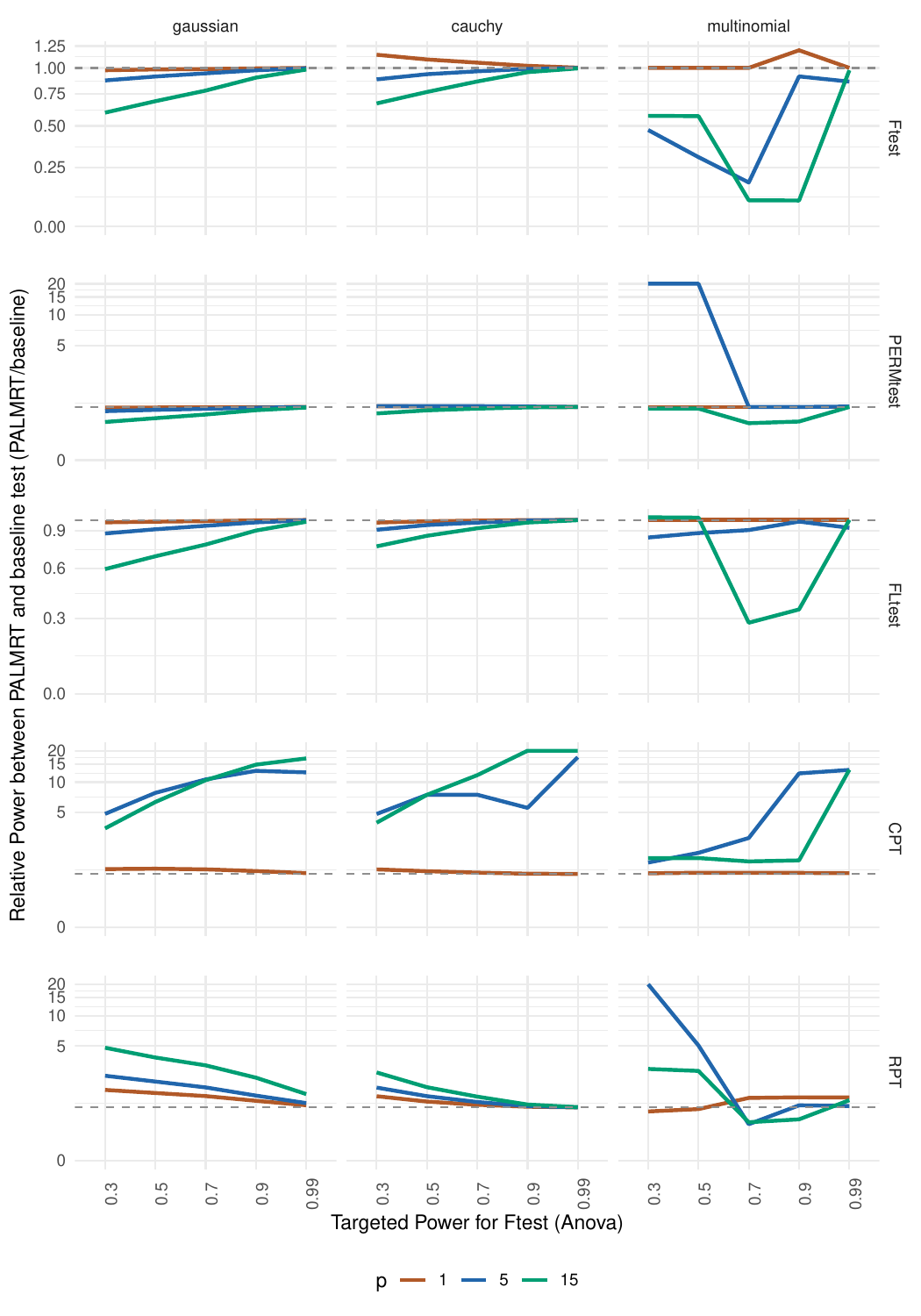}
\caption{Power analysis in Anova design, presented as line plots, organized by baseline methods (row names) and noise distributions (column names). Each line plot shows the median power ratio between PALMRT and a baseline method, plotted against the targeted F-test absolute power for various signal sizes. Different colors indicate varying feature dimensions $p$.}
\label{fig:powerC}
\end{center}
\end{figure}
In each experiment, we simulate data from the alternative setting by varying the linear coefficient in front of $x$ in eq.~(\ref{eq:sim_power}):
\begin{equation}
\label{eq:sim_power}
y = x\beta+\varepsilon.
\end{equation}
We set $\beta$ using Monte-Carlo simulation to yield F-test powers of approximately $30\%$, $50\%$, $70\%$, $90\%$, $99\%$. Figures \ref{fig:powerA}-\ref{fig:powerC} 
 display the median of relative power of PALMRT to F-test, PERMtest, FLtest, CPT and RPT for varying signal strengths, noise distributions and feature dimensions under Gaussian, Cauchy, and ANOVA designs. For visualization purposes, if a median ratio is greater than 20,  it is truncated at 20.

Under the Gaussian design, Ftest, PERMtest, and FLtest are most powerful across various noise types and feature dimensions.   PALMRT closely rivals them at $p=1$ and its relative sensitivity to low signal strength decreases with increased feature dimensions.  PALMRT achieves higher power than  CPT at $p=1$ and substantially outperforms CPT at $p=5$ and $p=15$, despite being more empirically conservative for controlling the type I error.  The gap between CPT and PALMRT does not disappear for $p=5$ and $p=15$ as we increase the signal strength. This is very different from the group bound method \citep{meinshausen2015group} which is also conservative but is powerless compared to CPT in a wide range of signal-to-noise ratios when examined in \citet{lei2021assumption}. PALMRT also shows much higher power compared to RPT with different concomitant feature dimension $p$, especially among the regime with a low signal-to-noise ratio.

In the Cauchy design, the relative sensitivities of PALMRT trend similarly to the Gaussian setting. However, F-test loses power with heavy-tailed noise distributions like Cauchy or multinomial, and random permutation tests can outperform F-test at low signal strengths. In the ANOVA design, results align with the Cauchy design when the noise is Gaussian or Cauchy. When the noise is multinomial, although the overall pattern is difficult to characterize, CPT and RPT are both significantly worse than PALMRT.

Overall, PALMRT has consistently higher power than CPT  in the Gaussian, Cauchy, and ANOVA designs, despite becoming more conservative as  $p$ increases, and is also more powerful than   RPT. When  $(n\slash p)$ is large, PALMRT and FLtest perform similarly, but their relative diverges as $(n\slash p)$ decreases. This divergence is contingent on both design and noise distributions.


\subsection{Coverage evaluation of $\rm{CI}_{\alpha}$}
\label{subsec:sim_CI}
\begin{figure}
\includegraphics[width = 1\textwidth]{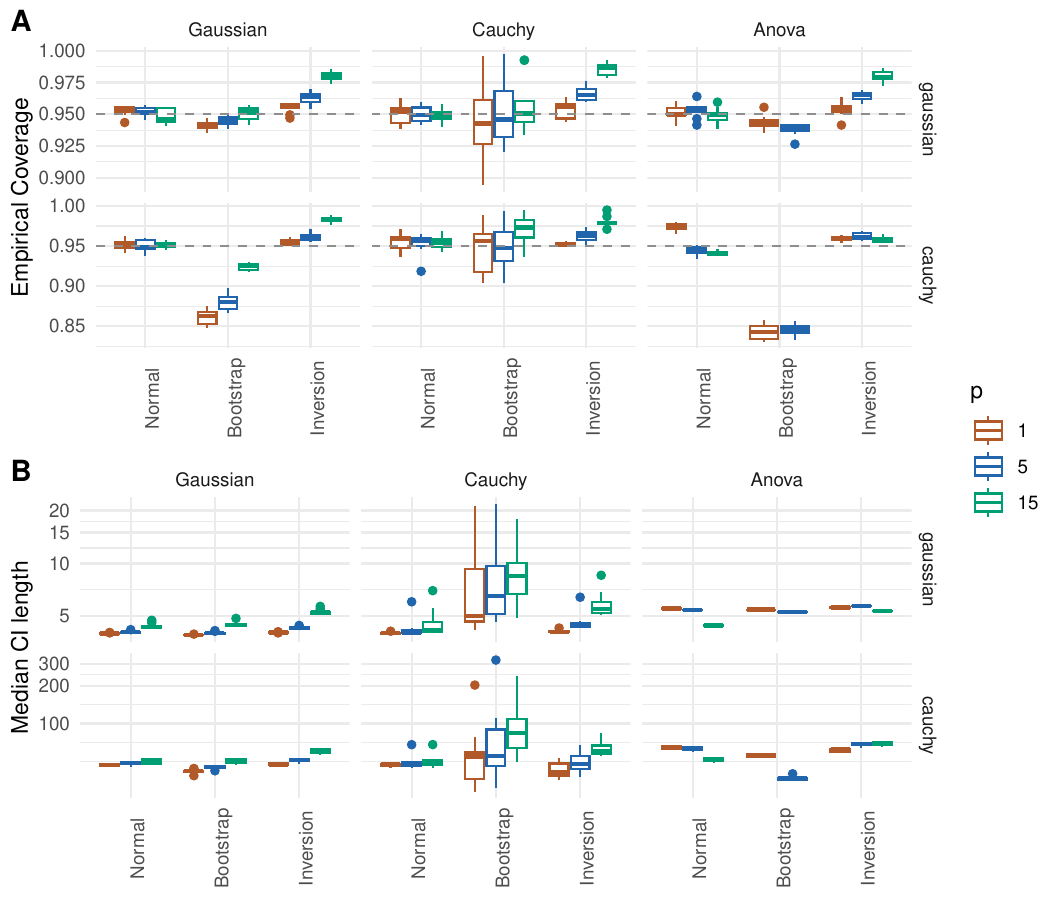}
\caption{CI coverage and length comparisons, presented as boxplots, organized by designs (row names) and noise distributions (column names).  Panel A displays the boxplots of coverage (y-axis) for different CI construction methods (x-axis). Panel B shows the boxplots of median CI length (y-axis) for different CI construction methods (x-axis). Both panels are colored by the feature dimension $p$.}
\label{fig:sim_CI}
\end{figure}
In this section, we evaluate the empirical coverage and median length of confidence intervals (\( \rm{CI}_{\alpha} \)) constructed using Algorithm \ref{alg:CI} (``Inversion"), Bootstrap, and normal approximation (``Normal") across Gaussian, Cauchy, and Anova designs for various \( \beta \). We exclude the Paired design due to undefined Bootstrap CIs for all $p$. As CI coverage and length are consistent across different \( \beta \), for the sake of space, we focus on results with a targeted F-test power of \( 50\% \) and defer full results to the supplementary materials.

Figure \ref{fig:sim_CI}A presents the achieved coverage. The CIs from Inversion consistently meet the desired coverage. In contrast, the normal CIs exhibit slight under-coverage in the ANOVA design with Cauchy noise, and Bootstrap CIs show severe under-coverage for Cauchy noise. For the ANOVA design with \( p=15 \), Bootstrap CIs are undefined and thus omitted in Figure \ref{fig:sim_CI}A. Figure \ref{fig:sim_CI}B shows the median CI lengths for each method. In line with previous findings, CIs from Algorithm \ref{alg:CI} are generally wider than Normal CIs, except for \( p=1 \) with Cauchy noise. Bootstrap CIs can exhibit greater variability when the designs are dominated by extreme values, such as in the Cauchy design setting.

\section{Robust Identification of long-Covid biomarkers}
\begin{figure}
\begin{center}
\includegraphics[width = .98\textwidth]{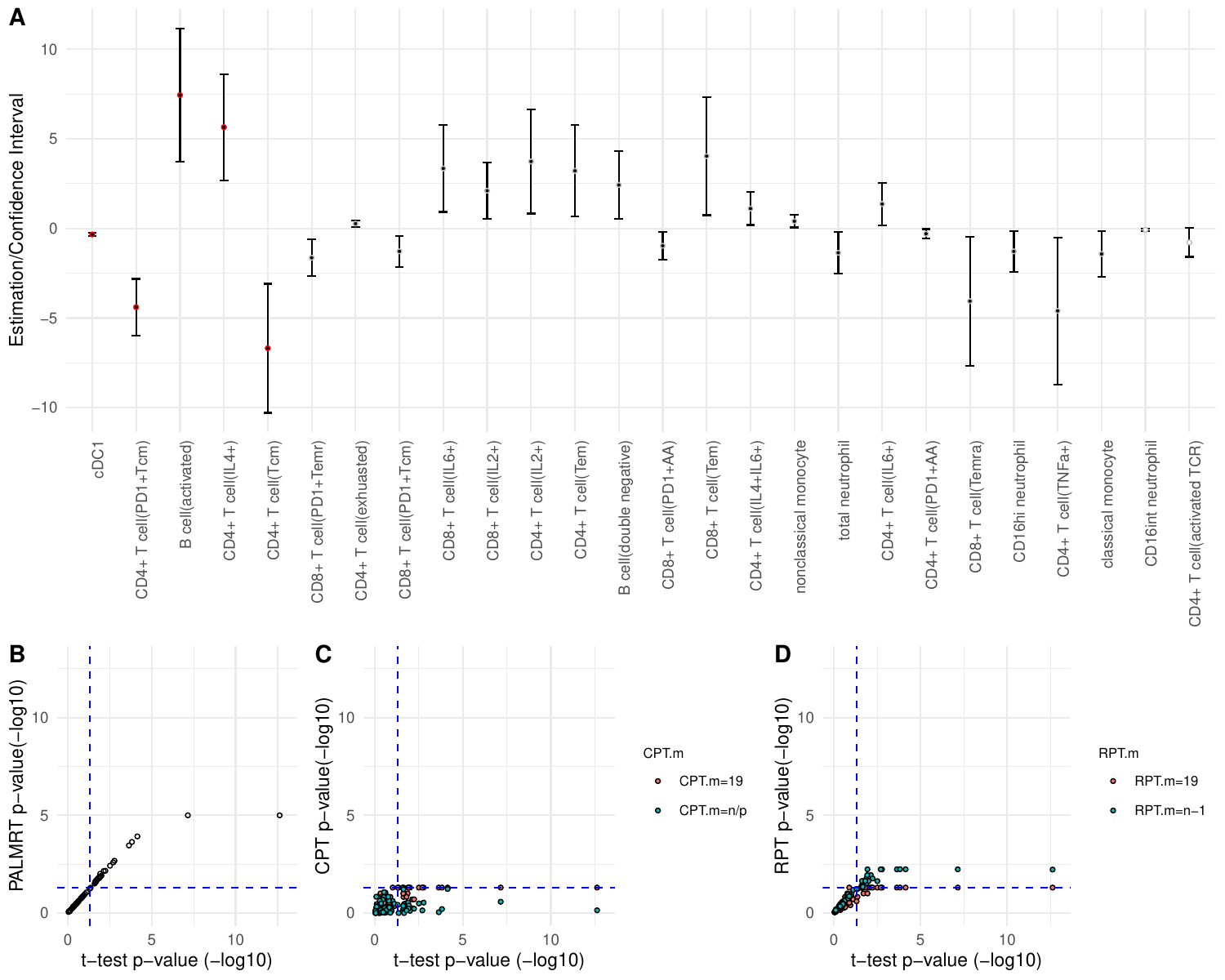}
\caption{Biomarker discovery using PALMRT and CPT. Panel A displays the estimated coefficient (center dot) and constructed CI (segmented line)for PALMRT for the 26 significant t-test findings before correction, with the x-axis label showing the feature name (ordered based on significance). The center dot is solid if p-value$\leq \alpha$ and empty otherwise using PALMRT. If a feature is significant after correction using PALMRT, the center dot is further highlighted with red circle. Panel B shows the negative log10 of the PALMRT p-value against that from the t-test for 64 features. Panel C shows the negative log10 of the CPT p-value against that from the t-test for 64 features, where a dot is colored red if CPT used $m = 19$ (CPT.m$=19$) and blue if CPT used $m = \lfloor \frac{n}{p}\rfloor-1$ (CPT.m$=n\slash p$). Panel D shows the negative log10 of the RPT p-value against that from the t-test for 64 features, where a dot is colored red if RPT used $m = 19$ (RPT.m$=19$) and blue if RPT used $m =n-1$ (RPT.m$=n-1$). The vertical\slash horizontal dashed lines represent the p-value level of $0.05$ in Panels B-D.}
\label{fig:MYLC}
\end{center}
\end{figure}
The MY-LC dataset contains measurements of 64 cell frequencies for 101 long-Covid (LC) participants and 84 controls (42 healthy samples, 42 convalescent samples without LC). A small percentage of measures are missing, with the number of observed samples ranging from 169 to 177 across features. Significant partial correlations between LC status and cell frequencies were identified by \citet{klein2022distinguishing}, after controlling for age, sex, and BMI as described in eq.\;(\ref{eq:MYLC}). Among the 64 features, 26 had \( p \leq 0.05 \), and 5 survived multiple hypothesis correction using the BH procedure. In this work, we apply PALMRT, CPT, and RPT,  which are theoretically guaranteed and validated through simulations, for robust biomarker identification. 

Figure \ref{fig:MYLC}A displays confidence intervals generated using Algorithm \ref{alg:CI} for the 26 significant biomarkers identified by the t-test, before multiple corrections, at \( \alpha = 0.05 \). The center dot in each CI represents the estimated LC coefficient in eq.\;(\ref{eq:MYLC}). Solid dots indicate features that were significant before correction; empty dots indicate otherwise.  PALMRT confirmed 24 of these 26 biomarkers. In particular, all 5 top biomarkers, which were significant after correction using the t-test, remained significant after correction using PALMRT (solid dot with red circle). In general, the p-values from PALMRT and the t-test are highly concordant, with ratios between them for the same feature ranging from 0.6 to 1.2, except for cDC1 and the central memory CD4+ T cell with positive PD1 (CD4+ T cell (PD1+Tcm)), which achieved the smallest possible PALMRT p-values at \( B = 10^5 \) and also had smallest t-test p-values $\ll 10^{-5}$ (see Figure \ref{fig:MYLC}B).

Both CPT and RPT, in contrast, reduced discoveries before correction and yielded none afterward. Increasing the number of cyclic permutations ($m$) from $m = 19$ to $m = \lfloor \frac{n}{p} \rfloor - 1$ did not improve their power in CPT, as shown in Figure~\ref{fig:MYLC}C. The nominal p-value cutoff at 0.05 resulted in 16 out of 26 t-test discoveries being confirmed when $m = 19$ and none being significant when $m = \lfloor \frac{n}{p} \rfloor - 1$ (see Figure~\ref{fig:MYLC}C). RPT achieved higher power compared to CPT at the nominal p-value cutoff of 0.05, confirming 22 out of 26 t-test discoveries when $m=19$. However, even for RPT, no discoveries were retained after correcting for multiple hypothesis tests, whether using $m=19$ or $m = n-1$, the largest $m$ allowed by the RPT algorithm (see Figure~\ref{fig:MYLC}D).

\section{Discussions}
We introduce a novel conformal test, PALMRT, designed for hypothesis testing in linear regression. This test and its corresponding confidence intervals are efficiently computed by evaluating statistic pairs, which are formed by augmenting the original regression problem with row-permuted versions of  \((x, Z)\). PALMRT achieves little power loss compared to conventional tests like the F/t-test and FL-test, and differs from CPT and RPT which also enjoy a worst-case coverage guarantee. Unlike CPT, PALMRT eliminates the need for complex optimization to construct the test and consistently outperforms CPT in both simulated and real-data scenarios. In comparison to RPT, another recently proposed method, developed independently to address the challenges in CPT, PALMRT does not require specially designed permutations that form a group, and its performance remains stable and does not deteriorate as the number of permutations increases.

Of note, the improvement over CPT is not universal. For example, in special designs like the paired design where there are many duplicate rows in \(Z\), CPT can be more powerful. Figure \ref{fig:powerD} illustrates the comparative power of CPT and PALMRT in signal detection under the paired design. While PALMRT generally matches or exceeds the power of CPT when \(p=1, 5\), it fails to detect signals at \(\alpha = 0.05\) when \(p=15\). CPT, conversely, successfully identifies an effective pre-ordering and \(\eta\) to construct a non-trivial test. While such settings are rare in practice, this observation raises an intriguing theoretical inquiry: Can the power of PALMRT be enhanced by filtering the random permutations \(\{x_{\pi_b}, Z_{\pi_b}\}\) to avoid near co-linearity among \(x, Z, x_{\pi_b}, Z_{\pi_b}\)? We earmark this question for future exploration.
\begin{figure}
\begin{center}
\includegraphics[width = 1\textwidth]{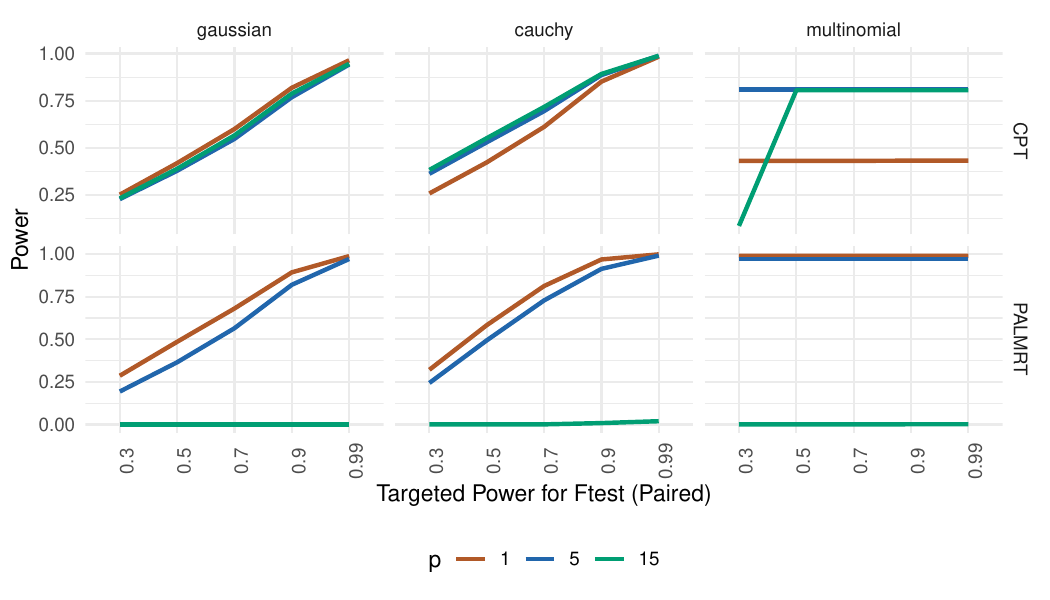}
\caption{Power analysis in Anova design, presented as line plots, organized by  methods (row names) and noise distributions (column names). Each line plot shows the median power of a given method and dimension $p$, plotted against the targeted F-test power for various signal sizes. Different colors indicate varying feature dimensions $p$.}
\label{fig:powerD}
\end{center}
\end{figure}

As a brief detour from our main discussion, the augmentation step in PALMRT is reminiscent of the seminal work by \citet{barber2015controlling}, which introduced the concept of knockoffs. This approach generates a knockoff copy \(\tilde{X}\) of the original feature matrix \(X\), e.g., \(X \leftarrow (x, Z)\) in our notation with dimension \(n \times (p+1)\). The key requirement of the knockoff copy is that swapping any pair \((\tilde{X}_j, X_j)\) leaves the covariance structure unchanged. Consequently, important quantities in regression analysis---such as OLS or Lasso coefficients \((\hat{\beta}_1, \ldots, \hat{\beta}_{p+1}, \tilde{\beta}_1, \ldots, \tilde{\beta}_{p+1})\)---are invariant under these swaps, provided the error terms \(\varepsilon_1, \ldots, \varepsilon_n\) are i.i.d Gaussian. Knockoffs control the FDR by retaining features \(X_j\) for which \(|\hat{\beta}_j| - |\tilde{\beta}_j|\) is sufficiently large. Although both PALMRT and knockoffs operate under a fixed design and have similar sample size requirements, their objectives differ. Knockoffs aim to control FDR across \(p+1\) features under the Gaussian noise assumption, often in more complex model-fitting contexts, whereas PALMRT focuses on computing individual p-values for partial correlations under a more relaxed exchangeability condition for the residual errors. Recent advancements in derandomized knockoffs \citep{ren2022derandomized} allow for the computation of modified e-values \citep{vovk2021values,wang2022false} through repeated runs with different \(\tilde{X}\) copies. However, achieving small p-values or high e-values necessitates a large \(n > \mbox{e-value}\) or \(n > (1\slash \mbox{p-value})\) at least, a constraint not shared by PALMRT. This makes PALMRT particularly advantageous for exploratory analyses aimed at uncovering partial correlations among a potentially large set of response and primary feature pairs, after adjusting for a limited number of covariates.  

\section*{Acknowledgment}
The author would like to thank Dr. Akiko Iwasaki and her team for providing access to the MY-LC data.
\section*{Funding}
 L.G. was supported in part by the NSF award DMS2310836. 

\bibliographystyle{chicago}
\bibliography{test}


\appendix 
\section{Proofs}
\label{app:proof}
\subsection{Proof of Lemma \ref{lem:CI_term}}
\begin{proof}
Recall that $T_{0b}(\beta)=\|(I - H^{x_{\pi_b}z_{\pi_b}z})(y-x\beta)\|_2^2$ and $T_{b0}(\beta)=\|(I - H^{xzz_{\pi}})(y-x\beta)\|_2^2=\|(I - H^{xzz_{\pi}})y\|_2^2=c_{b4}$. Hence, $T_{b0}(\beta)$ does not depend on $\beta$, and we only need to examine how $\beta$ influences $T_{0b}(\beta)$. When $x$ lies in the space spanned by $(z, z_{\pi_b}, x_{\pi_b})$, we have $c_{b1} = 0$, and
\begin{align*}
T_{0b}(\beta)=&\|(I-H^{x_{\pi_b}zz_{\pi_b}})(y-x\beta)\|_2^2\\
=&\|(I-H^{x_{\pi_b}z_{\pi_b}})y\|_2^2\\
=&\|(I-H^{x_{\pi_b}z_{\pi_b}z})y\|_2^2=c_{b3}.
\end{align*}
Hence, when $c_{b3} = c_{b4}$, we have $\omega_b(\beta) =1\slash 2$ and when $c_{b3} < c_{b4}$, we have $\omega_b(\beta) =1$.

When $x$ does not lie in the space spanned by $(z, z_{\pi_b}, x_{\pi_b})$, we have, $c_{b1} > 0$, and
\begin{align*}
 T_{0b}\leq T_{b0}\Leftrightarrow&\|(I-H^{x_{\pi_b}z_{\pi_b}z})(y-x\beta)\|_2^2\leq  c_{b4},\\
\Leftrightarrow& c_{b1}\beta^2 -2c_{b2}\beta +(c_{b3}-c_{b4})\leq 0,\\
\Rightarrow& s_b\coloneqq \frac{c_{b2}-\sqrt{c_{b2}^2-c_{b1} (c_{b3}-c_{b4})}}{a_{b1}}\leq \beta \leq u_b\coloneqq \frac{c_{b2}+\sqrt{c_{b2}^2-c_{b1} (c_{b3}-c_{b4})}}{a_{b1}}.
\end{align*}
Note that we must have $c_{b2}^2\geq c_{b1}(c_{b3}-c_{b4})$, and consequently, real-valued solutions  $s_b\leq u_b$ exist.  This follows from the following argument.

Set $\beta$ as the OLS coefficient estimate for $x$ in the regression $y\sim x+Z+x_{\pi}+Z_{\pi}$. We have
\begin{align*}
T_{0b}=&\|(I-H^{x_{\pi_b}z_{\pi_b}z})(y-x\beta)\|_2^2\\
=&\|(I-H^{xzx_{\pi_b}z_{\pi_b}})y\|_2^2\\
=&\|(I-H^{R^{x_{\pi_b}|xzz_{\pi_b}}})(I-H^{xzz_{\pi_b}})y\|_2^2\\
\leq & \|(I-H^{xzz_{\pi_b}})y\|_2^2=T_{b0}.
\end{align*}
Here, we use $R^{x_{\pi_b}|z_{\pi_b}xz}$ to denote the regression residuals of $x_{\pi_b}$ in the regression $x_{\pi}\sim Z+x+Z_{\pi}$.  Hence, the previous quadratic function is guaranteed to have real-valued root(s) when $c_{b1}>0$.
\end{proof}
\subsection{Proof of Theorem \ref{thm:CI}}
Utilizing Lemma \ref{lem:CI_term}, we know that value $\beta$ satisfies $f(\beta)\coloneqq \frac{1+\sum_{b=1}^B \omega_b(\beta)}{B+1} > \alpha$ if and only if
\[
f_{A_1}(\beta)\coloneqq \sum_{b\in A_1}\omega_b(\beta) > \gamma \coloneqq (B+1)\alpha - |A_2|-\frac{1}{2}|A_3|-1,
\]
where  \(A_1=\{b: c_{b1}>0\}\), \(A_2 =  \{b: c_{b1}=0, c_{b3} < c_{b4}\}\), and \(A_3 =\{b: c_{b1}=0, c_{b3} = c_{b4}\}\); $\omega_b(\beta) = \frac{1}{2}\mathbbm{1}\{s_b\leq \beta\leq u_b\}+ \frac{1}{2}\mathbbm{1}\{s_b< \beta< u_b\}$ for $b\in A_1$.   We can rearrange $\omega_b(\beta)$ for $b\in A_1$:
\begin{equation}
\label{eq:CI_proof1}
\omega_b(\beta) = \frac{1}{2}\left(\mathbbm{1}\{\beta\geq s_b\}+\mathbbm{1}\{\beta>s_b\}-\mathbbm{1}\{\beta\geq u_b\}-\mathbbm{1}\{\beta>u_b\}\right).
\end{equation}
Summing over all terms $b\in A_1$, we obtain that
\begin{align}
f_{A_1}(\beta)&=\frac{1}{2}\left(\sum_{b\in A_1} \mathbbm{1}\{\beta\geq s_b\}+\sum_{b\in A_1} \mathbbm{1}\{\beta>s_b\}-\sum_{b\in A_1} \mathbbm{1}\{\beta\geq u_b\}-\sum_{b\in A_1} \mathbbm{1}\{\beta>u_b\}\right)\notag\\
& = \frac{1}{2}\left(\#\{b:s_b \leq \beta\}+\#\{b:s_b < \beta\}-\#\{b:u_b \leq \beta\}-\#\{b:u_b < \beta\}\right). \label{eq:CI_proof2}
\end{align}
The function  $f_{A_1}(\beta)\geq 0$ is a piece-wise constant function with change points being $t_1 < \ldots < t_M$.   Hence, when $\gamma < 0$, $f_{A}(\beta)>\gamma$ is guaranteed.

When $\gamma\geq 0$, we must have $t_1\leq \beta < t_M$ in order to have $f_A(\beta)>\gamma$ since $f(\beta) = 0$ for  $\beta < t_1$ and $\beta \geq t_M$.  It is obvious that for $\beta = t_1$,  $f_A(\beta)=\frac{1}{2}(m_1^s-m^u_1)$. Suppose we know $f_{A}(t_l)$. Then,
\begin{itemize}
\item  Increasing $\beta$ from $t_l$ to a value in $(t_l, t_{l+1})$: For all $\beta\in (t_l, t_{l+1})$, because for those $s_b$, $u_b$ smaller than $t_l$ or those no smaller than $t_{l+1}$, their contributions will not change in eq.~(\ref{eq:CI_proof2}), and only those at exactly at $t_l$ will add contribution $1\slash 2$ or $-1\slash 2$ as the corresponding comparisons to $t_l$ change from equality to inequality. Hence, we have
\[
f(\beta) = f_A(t_l)+\frac{1}{2}(m_l^s-m_l^u),\quad \forall \beta \in  (t_l, t_{l+1}).
\]
We denote this function value as $f(t_{l}^+)$ to indicate being in the open interval $(t_{l}, t_{l+1})$.
\item Increasing $\beta$ from a value in $(t_l, t_{l+1})$ to $t_{l+1}$: For those $s_b$, $u_b$  no greater than $t_{l}$ or those greater than $t_{l+1}$, their contributions will not change in eq.~(\ref{eq:CI_proof2}), and only those who now hit $t_{l+1}$ exactly will add contribution $1\slash 2$ or $-1\slash 2$ for newly satisfied comparisons to $t_{l+1}$. Hence, we have
\[
f(t_{l+1}) = f_A(t_l^+)+\frac{1}{2}(m_{l+1}^s-m_{l+1}^u).
\]
\end{itemize}
These confirms our induction rule described in the main paper. Finally,
\begin{itemize}
\item  The smallest value of $\beta$ at the critical values $\{t_l\}_{l=1}^M$ that satisfies $f_A(t_l)>\gamma$ is $\min\{t_l: f_A(t_l) > \gamma\}$. The smallest interval $(t_l, t_{l+1})$ that satisfies $f_A(\beta)>\gamma$ for $\beta\in (t_l, t_{l+1})$ is $\min\{(t_l, t_{l+1}): f_A(t_l^+) > \gamma\}$, which leads to infimum value $\min\{t_l: f_A(t_l^+) > \gamma\}$. Hence, $\beta_{\min}=\min\{t_l: f_A(t_l)\vee f_A(t_l^+) > \gamma\}$.
\item The largest value of $\beta$ at the critical values $\{t_l\}_{l=1}^M$ that satisfies $f_A(t_l)>\gamma$ is $\max\{t_l: f_A(t_l) > \gamma\}$. The largest interval $(t_l, t_{l+1})$ that satisfies $f_A(\beta)>\gamma$ for $\beta\in (t_l, t_{l+1})$ is $\max\{(t_l, t_{l+1}): f_A(t_l^+) > \gamma\}$, which leads to supremum value $\max\{t_l: f_A(t_{l-1}^+) > \gamma\}$. Hence, $\beta_{\max}=\max\{t_l: f_A(t_l)\vee f_A(t_{l-1}^+) > \gamma\}$.
\end{itemize}
This conclude our proof that Algorithm \ref{alg:CI} constructs the exact $\rm{CI}_\alpha$ described in Lemma \ref{lem:CI_term}.
\section{Additional simulation results}
\label{app:sim}
\subsection{Type I error control at $\alpha = 0.01, 0.001$}.
\begin{figure}[H]
\begin{center}
\includegraphics[width = .8\textwidth]{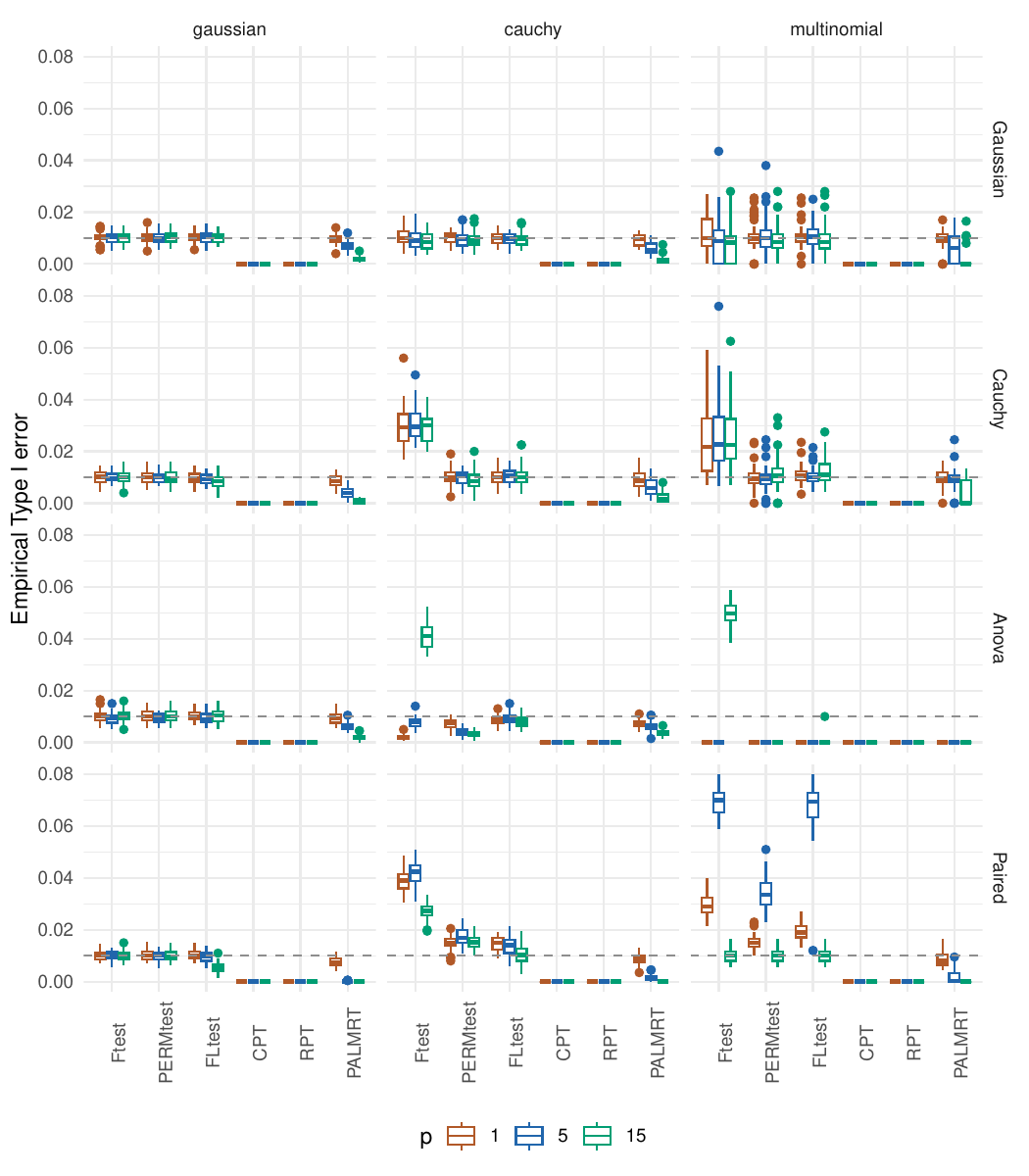}
\caption{Empirical type I error using various methods, organized into boxplots by corresponding designs (row names) and noise distributions (column names). Each boxplot displays the empirical type I error for different methods, separately for different feature dimension $p$ (color). The dashed horizontal line represents the targeted type I error  $\alpha = 0.01$.}
\label{fig:typeIB}
\end{center}
\end{figure}

\begin{figure}[H]
\begin{center}
\includegraphics[width = 1\textwidth]{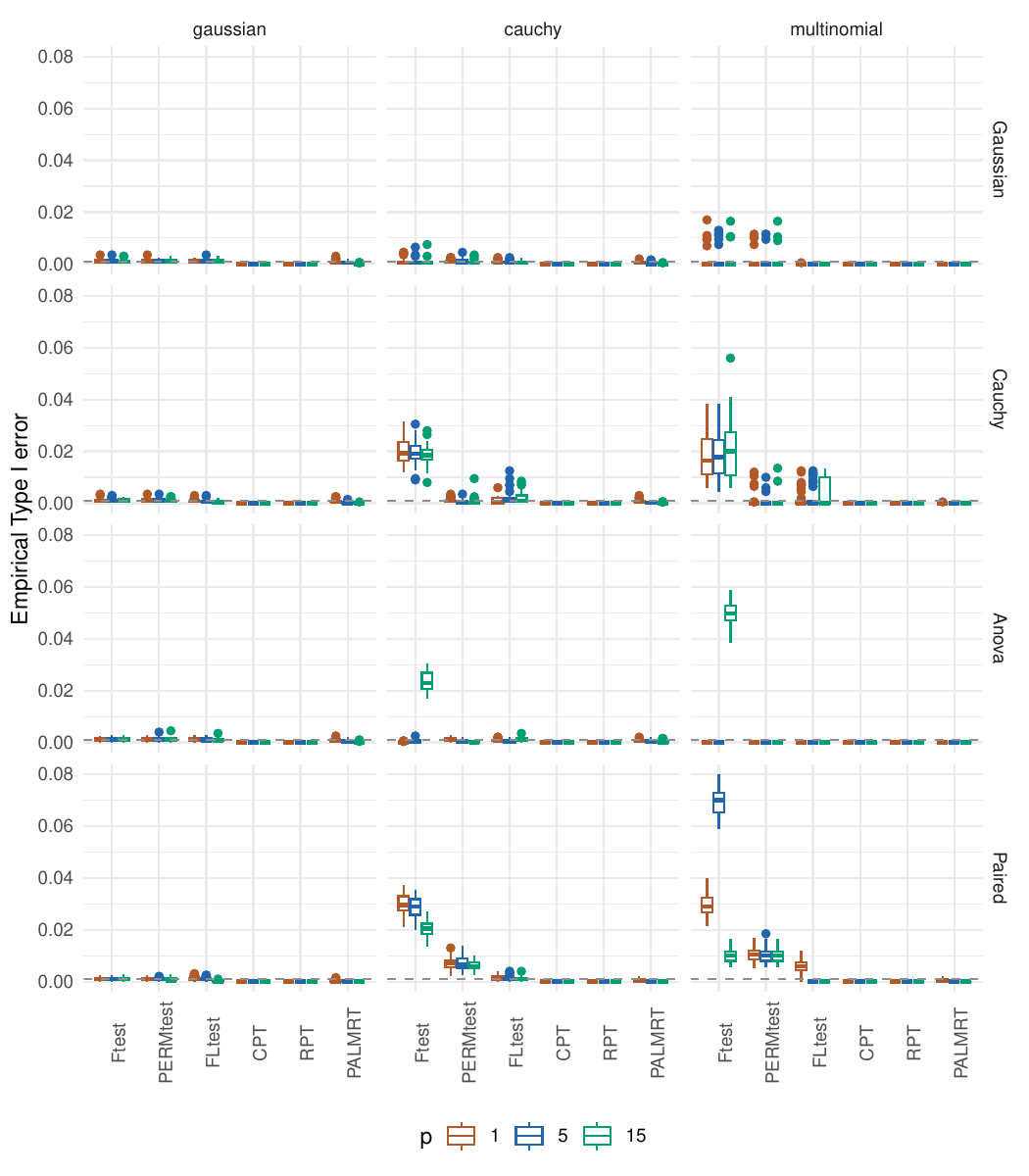}
\caption{Empirical type I error using various methods, organized into boxplots by corresponding designs (row names) and noise distributions (column names). Each boxplot displays the empirical type I error for different methods, separately for different feature dimension $p$ (color). The dashed horizontal line represents the targeted type I error  $\alpha = 0.01$.}
\end{center}
\label{fig:typeIB}
\end{figure}

\subsection{CI coverage with varying signal sizes}.
\begin{figure}[H]
\includegraphics[width = .9\textwidth]{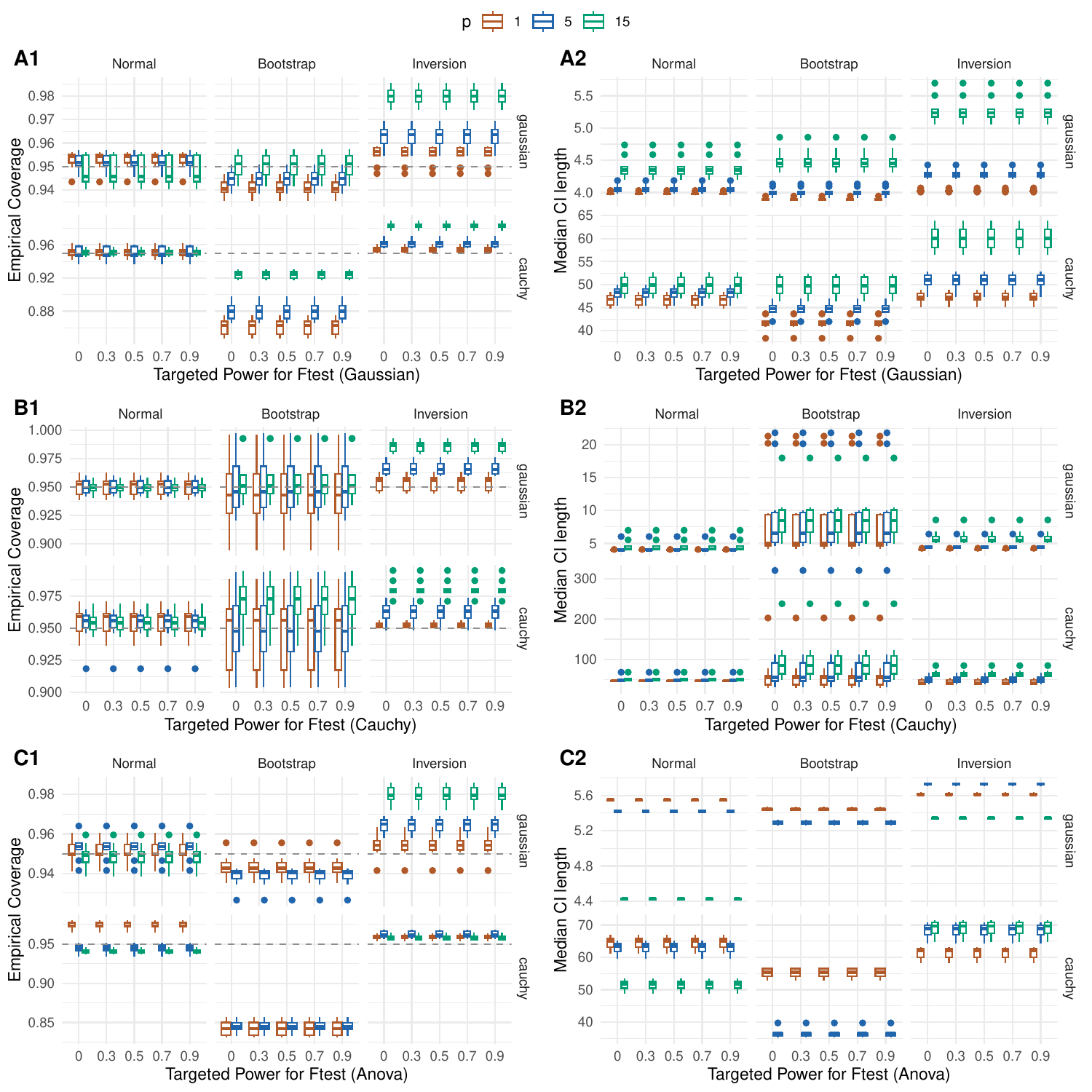}
\caption{CI coverage and length comparisons, presented as boxplots.  Panels A1-C1 display the boxplots of coverage (y-axis) for different signal strength (x-axis) for different designs (indicated in the x-lab label), different noises (row names of a panel), methods (column names of a panel), and feature dimension $p$ (color). Panels A2-C3 display the boxplots of  CI's median length (y-axis) for different signal strength (x-axis) for different designs (indicated in the x-lab label), different noises (row names of a panel), methods (column names of a panel), and feature dimension $p$ (color). }
\label{fig:sim_CI2}
\end{figure}

\end{document}